\documentclass[journal]{IEEEtran}
\usepackage{svg}
\usepackage{cite}
\usepackage{amsmath,amssymb,amsfonts,balance}

\usepackage[algo2e,ruled,linesnumbered,lined,boxed,commentsnumbered]{algorithm2e}
\usepackage[
    colorlinks=true,
    linkcolor=blue,
    citecolor=blue,
    urlcolor=magenta
]{hyperref}
\usepackage{booktabs} 
\usepackage{multirow} 
\usepackage{array}    
\usepackage{algorithmic}
\usepackage{algorithm}
\usepackage{graphicx}
\usepackage{textcomp}
\usepackage{xcolor}
\usepackage{float}
\usepackage{placeins} 
\usepackage{caption}
\usepackage{amsthm}

\newtheorem{theorem}{Theorem}[section]

\newtheorem{lemma}[theorem]{Lemma}
\newtheorem{proposition}[theorem]{Proposition}
\usepackage{subcaption}
\def\BibTeX{{\rm B\kern-.05em{\sc i\kern-.025em b}\kern-.08em
    T\kern-.1667em\lower.7ex\hbox{E}\kern-.125emX}}

\begin{document}

\title{\fontsize{24.88}{28}\selectfont Tensor-Based Dynamic Channel Estimation for mmWave Movable Antenna MIMO Systems}

\author{Zhendong Li, Linchu Chen, Lin Chen, Zhou Su, Ruoyu Zhang, Guangji Chen, Ying Wang, and Wen Chen
\thanks{Zhendong Li and Linchu Chen are with the School of Information and Communication Engineering, Xi’an Jiaotong University, Xi’an 710049, China (email: lizhendong@xjtu.edu.cn; chenlinchu@stu.xjtu.edu.cn). Lin Chen is with the Department of Electrical and Computer Engineering, Stevens Institute of Technology, Hoboken, NJ 07030, USA (e-mail: lchen53@stevens.edu). Zhou Su is with the School of Cyber Science and Engineering, Xi'an Jiaotong University, Xi'an 710049, China (email: zhousu@ieee.org). Ruoyu Zhang and Guangji Chen are with the School of Electronic and Optical Engineering, Nanjing University of Science and Technology, Nanjing 210094, China (e-mail: ryzhang19@njust.edu.cn; guangjichen@njust.edu.cn). Ying Wang is with the State Key Laboratory of Networking and Switching Technology, Beijing University of Posts and Telecommunications, Beijing 100876, China (wangying@bupt.edu.cn). Wen Chen is with the Department of Electronic Engineering, Shanghai Jiao Tong University, Shanghai 200240, China (e-mail: wenchen@sjtu.edu.cn). (Corresponding author: Zhou Su)}
\vspace{-2em}}


\maketitle

\begin{abstract}
This paper investigates the dynamic channel estimation algorithm in mmWave movable antenna (MA) multiple-input multiple-output (MIMO) systems. To achieve highly accurate channel estimation, we propose a tensor decomposition-based channel estimation algorithm. First, by leveraging the path response model and utilizing the intrinsic sparsity of mmWave channels, the channel corresponding to MA pairs at the base station and mobile station is transformed into a superposition of channels from sparse paths. Next, the received signal is constructed as a fourth-order tensor to fully capture the high-dimensional structural information of the MA MIMO channel. Then, two tensor decomposition schemes are adopted to extract the factor matrices,  and our analysis reveals that the uniqueness of the decomposition can be guaranteed in our model. Subsequently, the  propagation loss, frequency offset, angle of arrival/departure, and time delay are obtained based on these factor matrices and the channel matrix can be rebuilt. Additionally,   Cramér-Rao bound (CRB) is also derived as a performance evaluation standard, proving that the proposed algorithm achieves a higher estimation accuracy and nearly approaches this minimum bound. Moreover, normalized mean square error (NMSE) is selected as the  {evaluation metric} for estimation accuracy. Finally, simulation results reveal a notable reduction in the estimation error of the proposed algorithm when compared to the baseline algorithms, confirming its estimation advantage.
\end{abstract}

\begin{IEEEkeywords}
Dynamic channel estimation, movable antenna, path response, tensor decomposition, Cramér-Rao bound.
\end{IEEEkeywords}

\section{Introduction}
\IEEEPARstart {B}{y} using abundance of available high spectrum and spatial diversity, multiple-input multiple-output (MIMO) techniques substantially boosts achievable throughput in wireless networks thanks to its notable benefits such as substantial path loss mitigation, enhanced spectral efficiency, and improved energy efficiency, representing a key enabling technology for the future wireless networks\cite{8284058}. Nevertheless, the intrinsic geometrical limitations of fixed antenna arrays impose restrictions on the ability of MIMO systems to optimize wireless channel degrees of freedom (DoFs)\cite{10403776, bajianle2}. This deficiency remains existing in massive MIMO with substantial numbers of static antenna elements, whose deployment comes with the expense of significantly elevated hardware expenditures and energy requirements owing to the growing quantities of radio frequency (RF) chains and antennas \cite{1284943}. To overcome these limitations, movable antenna (MA) is investigated as an efficient scheme to utilize the DoFs\cite{10753482}.

Compared to conventional fixed antenna systems, the MA MIMO systems have some appealing advantages by fully exploiting DoFs. First, MA MIMO systems enable comprehensive utilization of the spatial diversity available across a given spatial region through dynamic reconfiguration of its locations to lower the directional correlation between steering vectors, leading to decreased ambiguity and reduced interference effects. Then, MA MIMO systems  {can improve spatial} multiplexing performance. Through coordinated optimization of multiple antenna placements, the propagation characteristics between transmitter and receiver can be reconfigured to optimize the MIMO channel capacity\cite{10354003}. Moreover, MA can intelligently adjust their positions to seek improved channel states in contrast to the conventional fixed antenna systems which passively endure random channel fading. This capability allows MA MIMO systems to achieve a higher performance while maintaining or even reducing the  antenna count and requisite RF chains compared to conventional  fixed antenna systems\cite{10286328,11015925}. Therefore, to fully exploit these advantages, the optimization of antenna positions and beamforming emerges as a critical enabler which is critical for harnessing the complete potential of MA MIMO systems.

Nowadays, the antenna position optimization and beam forming problem in MA MIMO systems are extensively investigated, and these existing research can be classified based on the optimization goals, which include optimizing spectral efficiency\cite{10354003},  {maximizing channel capacity\cite{beamsum, 11015925, 10243545}}, enhancing energy efficiency\cite{10278220, 10318061, bajianle1}, and ensuring physical layer security\cite{11177504, TowedMA, zhaoyujie1}.   { An optimization  problem was formulated in \cite{10354003} to minimize the Cramér-Rao bound (CRB) through adjustment of the antenna position vector which comprised all MA coordinates.} The authors of \cite{10243545} presented a joint optimization framework for MA positioning to enhance both spectral efficiency and energy performance across diverse MIMO system configurations. In \cite{10278220}, a beamforming design problem was investigated by exploiting the DoF via antenna position optimization to obtain complete array gain along the target direction, with null steering enforced in all remaining undesired directions.  {Moreover, the position optimization scheme of towed MA for airborne communication systems was proposed in \cite{TowedMA} to enhance the physical layer security.} However, the optimization of MA positions in above mentioned works generally requires the knowledge of the precise channel state information (CSI). Hence, channel estimation must be performed prior to undertaking these optimization tasks.

Some current works have already been carried out regarding the issue of channel estimation in MA systems\cite{11012849,10497534, 11149143, 11087029, 10807122}. In particular,  a maximum likelihood channel estimation scheme was developed in \cite{11012849} by capitalizing on channel sparsity to achieve improved estimation performance. In \cite{10497534}, a general channel estimation framework based on compressed sensing (CS) for MA systems was proposed. Based on multiple channel measurements at designated positions, the angle parameters and complex coefficients of the multi-path components are jointly estimated by exploiting the multi-path field response channel structure.  In \cite{10807122}, a directional sparsity-aided least square (LS) algorithm was proposed for each local unit to estimate the instantaneous CSI. However, both LS and CS  {estimation algorithms neglect} the high-dimensional characteristics in the received signals, which results in unsatisfactory estimation results and high computational complexity. 

To address these issues, tensor decomposition-based channel estimation algorithm, which fully utilizes the high-dimensional property in mmWave MIMO systems, has attracted considerable research attention\cite{7914672, 10659325,10839618}.  A tensor decomposition-based framework for multi-path channel parameter extraction in MA systems was developed to achieve significant pilot overhead reduction and higher accuracy\cite{10659325}. Compared to traditional matrix-based algorithms, tensor decomposition-based algorithms can render the structural information of signals more explicitly and offer better performance in both computational complexity and estimation accuracy, thereby offering a new perspective for high-dimensional signal processing\cite{7914672}. Furthermore, parameter estimation based on tensor decomposition shows considerable robustness against noise, a property that significantly boosts its practical applicability in low signal-to-noise ratio (SNR) regimes\cite{7891546}. Nevertheless, most existing studies relied upon static channel models, overlooking the variations of CSI in dynamic scenarios.  {In practical communication, rapid relative motion may exist between the transceivers and the scatterers, which leads to frequency offsets in the received signal. This dynamicity is particularly pronounced in unmanned aerial vehicles and Internet of vehicles communications\cite{TutMAnet}. In such dynamic scenarios, optimizing antenna positions becomes even more challenging.} Therefore, it is essential to develop a dynamic channel model for mmWave MA MIMO systems which captures the time varying channel behavior, and more appropriate algorithms are desired to exploit the potential capability of MA systems.

Building upon the preceding analysis, this paper proposes a tensor decomposition-based dynamic channel estimation scheme for mmWave MA MIMO system. In specific, we represent the complex channel model as a superposition of channels over sparse paths by fully exploiting the spatial sparsity inherent in the propagation paths. Moreover, notice that received signal intrinsically exhibits high dimensional structure of MIMO system in dynamic scenarios, we reformulate the signal into a  {fourth-order} tensor. The execution of tensor decomposition enhances the discernibility of structural information across the signal's dimensions, consequently permitting the acquisition of more precise channel estimates. The main contributions are outlined as follows:

\begin{itemize}
\item  {We present a dynamic channel estimation framework for mmWave MA MIMO systems. Specifically, by leveraging the path response model and exploiting the inherent sparsity of mmWave channels, the channel associated with each MA pair at the base station (BS) and mobile station (MS) is expressed as a superposition of a limited number of propagation paths. This framework exploits the sparsity inherent in the channel, thereby avoiding the need to exhaustively estimate the channel between every pair of transmit and receive antennas. More importantly, the exploitation of this sparsity also provides the mathematical foundation for the tensor decomposition algorithm proposed in this paper.}

\item A tensor decomposition-based dynamic channel estimation scheme for mmWave MA MIMO systems is proposed.  {First, building on existing work, we further construct the received signal as a fourth-order tensor that is better suited to characterizing time-varying channels. The multi-dimensional characteristics of MIMO received signals and the low-rank structure of the mmWave channel are fully exploited. Subsequently, a channel parameter extraction method based on tensor decomposition scheme is developed. Furthermore, we provide an easily implementable uniqueness  {condition} for the proposed tensor decomposition approach by exploiting the Vandermonde structure. Finally, to assess the performance of the proposed algorithm, CRB for parameter estimation under this model is derived.}

\item  {The accuracy and effectiveness of the proposed algorithms are validated through extensive simulation experiments across various scenarios. The simulation results demonstrate their superior estimation accuracy compared with baseline algorithms. In addition, systematic analysis confirms their effectiveness under varying practical parameters, such as the number of time slots and subcarriers. The results further verify that the proposed algorithms can facilitate MA position optimization by providing reliable channel parameter estimates. To further validate the efficiency, we compare the runtime of different algorithms, confirming the lower computational cost of the proposed methods. Overall, the simulation results indicate that the two proposed algorithms are effective for different application needs and offer superior performance in its targeted setting.}
\end{itemize}

This paper is arranged in the following manner. Section II introduces the dynamic MA MIMO channel model and formulates the channel estimation problem. Section III represents the tensor decomposition-based channel estimation algorithm, including canonical polyadic (CP) decomposition methods, analysis of uniqueness conditions, subsequent parameter extraction procedures from factor matrices, derivation of CRB and assessment of computational complexity for proposed algorithm. Finally, numerical simulations are provided in Section IV, and the paper is concluded in Section V.

\textit{Notations}: The notation follows the convention of boldface lower-case, boldface upper-case, and calligraphic upper-case letters for vectors, matrices and tensors, as in \( a \), \( \mathbf{a} \), \( \mathbf{A} \), \( \mathcal{A} \), respectively. The \(r\)-th column vector of matrix \(\mathbf{A}\) is represented as \(\mathbf{a}_r\). The $i$-th entry of vector $\mathbf{a}$ is denoted by $\mathbf{a}_i$ and $(i,j)$-th entry of a matrix can be represented by $\mathbf{A}_{[i,j]}$, the same notation is adopted throughout this paper for describing both vectors and tensors. $\mathbf{A}_{[i:j, :]}$ represents the submatrix comprising rows $i$ to $j$ of matrix $\mathbf{A}$. Furthermore, \(\overline{\mathbf{A}}\) and \(\underline{\mathbf{A}}\) are derived by removing the top and bottom rows of \(\mathbf{A}\), respectively. \( \text{diag}(\mathbf{a}) \) denotes a diagonal matrix formed by \( \mathbf{a} \). Conjugation, transpose, Hermitian, pseudo-inverse, $\ell_2$-norm, and Frobenius norm are represented by \( (\cdot)^* \), \( (\cdot)^T \), \( (\cdot)^H \), \( (\cdot)^\dagger \), \( \|\cdot\| _2\), and \( \|\cdot\|_F \), in order. The symbols \(\otimes\), \(\odot\), \(*\), \(\cdot\), and \(\circ\) respectively denote the Kronecker, Khatri-Rao, Hadamard, inner, and outer products. The Kruskal-rank and the standard rank of a matrix \(\mathbf{A}\) are expressed as \(k(\mathbf{A})\) and \(\text{rank}(\mathbf{A})\). The operator \(\text{Re}\{\cdot\}\) extracts the real part of its complex argument. A zero-mean circularly-symmetric complex Gaussian distribution with variance \(\sigma^2\) is denoted by \(\mathcal{CN}(0, \sigma^2)\).

\section{System Model}
\subsection{mmWave MA MIMO System}
Fig. \ref{model} depicts the considered downlink mmWave MA MIMO system which  {consists of a BS and an MS}. The BS is deployed with $N_\text{BS}$ MAs while the MS is equipped with $N_\text{MS}$ MAs. Linear MA are deployed in both the BS and MS. We deploy $Q_\text{BS}$ and $Q_\text{MS}$ RF chains in BS and MS, respectively.  {The system utilizes $B$ bandwidth and $K_t$ orthogonal frequency division multiplexing (OFDM) subcarriers, with a subset of $K$ subcarriers being selected for dynamic channel estimation and are uniformly distributed across the operating bandwidth. The spacing between adjacent pilots is $\Delta_p$.} The entire transmission period is partitioned into $N$ successive time blocks and  {each} block is further divided into finer time slots.  The $n_s$-th OFDM symbol $\mathbf{s}_{k,m,n_s}\in \mathbb{C}^{d\times 1}$ on the $k$-th subcarrier and $m$-th time slot is digitally precoded, transformed to time domain via $K_t$-point IDFT,  {appended with a cyclic prefix, and finally processed by the RF precoder, which can be given as
\begin{equation}
\mathbf{x}_{k,m, n_s} = \mathbf{F}_\text{RF}\mathbf{F}_{k,m}\mathbf{s}_{k,m, n_s}.\label{fashexinhao}
\end{equation}
This hybrid precoding architecture can avoid the mutual coupling problem caused by antennas being too close to each other during movement\cite{huouxiaoying}.} To simplify the system design, the pilot signals are assumed to be identical across all subcarriers and time slots, i.e., $\mathbf{x}_{k,m,n_s} = \mathbf{x}_{n_s}$. Aggregating these signals column-wise yields the matrix $\mathbf{X} = [\mathbf{x}_{1}, \cdots, \mathbf{x}_{N_s}]$.

\FloatBarrier
\begin{figure}[ht]
    \centering    \includegraphics[width=\linewidth]{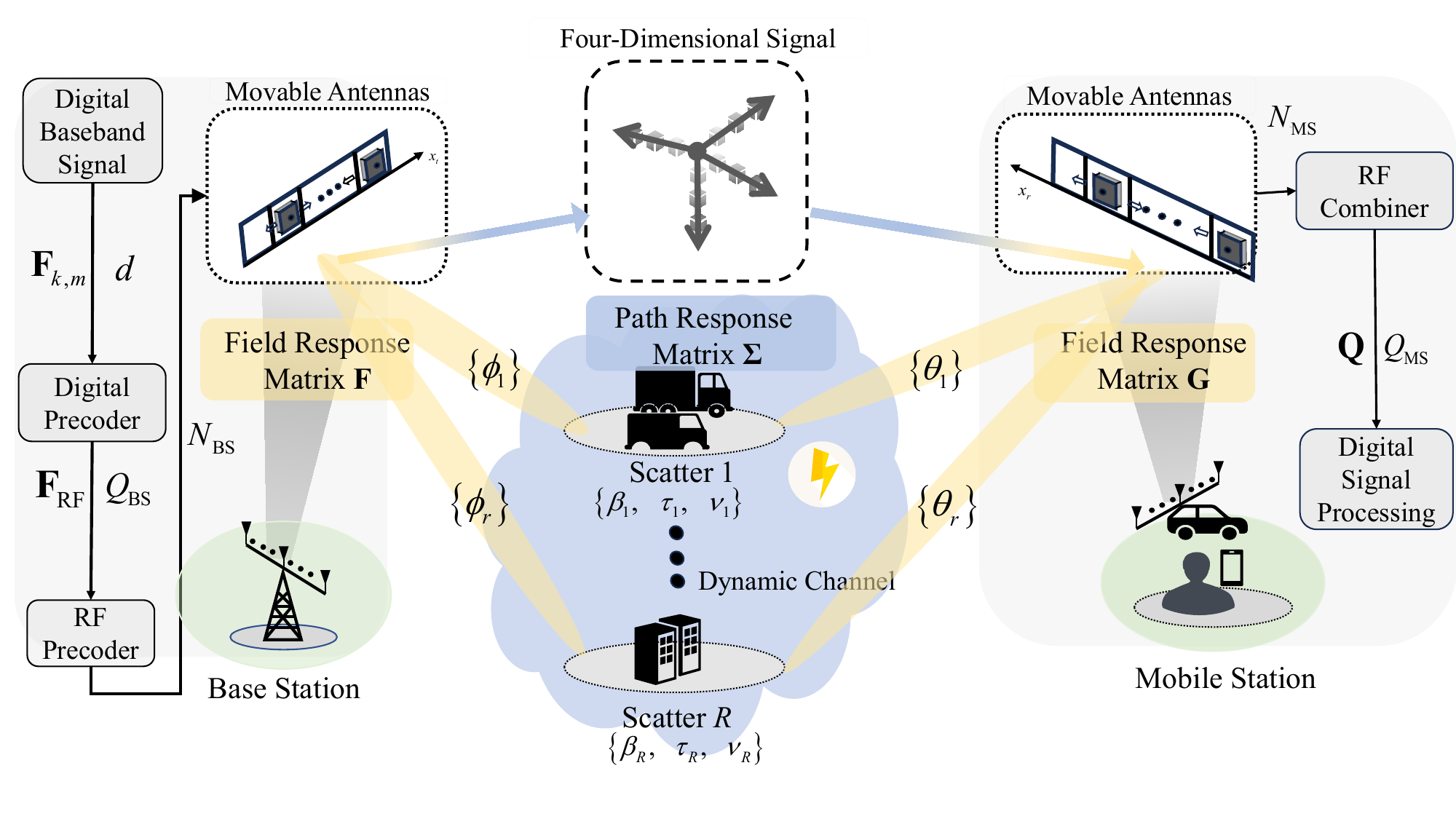}
    \caption{Illustration of the MA MIMO system with dynamic channel.}
    \label{model}
\end{figure}

 In MA systems, the positions of both transmitter-side and receiver-side antennas can be flexibly adjusted within their respective one-dimensional local regions, denoted as $\mathcal{R}_t$ and $\mathcal{R}_r$. For the $n_\text{BS}$-th  {transmitting} MA and $n_\text{MS}$-th  {receiving MA}, $(n_\text{BS} \in \{1, \cdots , N_\text{BS}\}, n_\text{MS} \in \{1, \cdots , N_\text{MS}\} )$, the positions of them can be specified by the coordinates $x^\text{t}_{n_\text{BS}} \in \mathcal{R}_t$ and $x^\text{r}_{n_\text{MS}}\in \mathcal{R}_r$, respectively. 
Let $\mathbf{t}_n = [x^t_1, . . . , x^t_{N_\text{BS}} ] \in\mathbb{R}^{1\times N_\text{BS}}$ and $\mathbf{r}_n = [x^r_1, . . . , x^r_{N_\text{MS}} ] \in\mathbb{R}^{1\times N_\text{MS}}$ be the collections for positions in the $n$-th time block of $N_\text{BS}$ antennas in BS and $N_\text{MS}$ antennas in MS, respectively. Thus, the steering vector in the mmWave MA MIMO system at the $n$-th time block can be given by
\begin{equation}
    \mathbf{f}(\mathbf{r}_n;\theta_r)=[e^{j\frac{2\pi}{\lambda}\rho(x^r_1; \theta_r)},\cdots,\ e^{j\frac{2\pi}{\lambda}\rho(x^r_{N_\text{MS}}; \theta_r)}]^T,\label{daoxiangshiliang1}
\end{equation}
\begin{equation}
    \mathbf{g}(\mathbf{t}_n;\phi_r)=[e^{j\frac{2\pi}{\lambda}\rho(x^t_1; \phi_r)},\cdots,\ e^{j\frac{2\pi}{\lambda}\rho(x^t_{N_{\text{BS}}}; \phi_r)}]^T,\label{daoxiangshiliang2}
\end{equation}
where $\rho$ represents the influence of the MA positions, for instance $\rho(x^r_1; \theta_r) = x_1^r\text{cos}\theta_r$ and  \(\lambda\) represents the carrier wavelength. Meanwhile, \(\theta_r\) and \(\phi_r\) correspond to the angle-of-arrival (AoA) and the angle-of-departure (AoD), respectively. By adjusting the antenna positions, MA MIMO systems can avoid deep fading and thereby optimize communication performance. The field response matrices of $N_\text{MS}$ transmitting MAs and $N_\text{BS}$ receiving MAs can be respectively written as
\begin{align}
        \mathbf{F}&=\left[\mathbf{f}(\mathbf{r}_n;\theta_1),\cdots ,\mathbf{f}(\mathbf{r}_n;\theta_R)\right]\in\mathbb{C}^{N_\text{MS}\times R},\\
        \mathbf{G}&=\left[\mathbf{g}(\mathbf{t}_n;\phi_1),\cdots,\mathbf{g}(\mathbf{t}_n;\phi_R)\right]\in\mathbb{C}^{N_\text{BS}\times R}.
\end{align}
The received signals first take processing by the common RF combiner. After the removal of the cyclic prefix, a $K_t$-point discrete Fourier transform (DFT) is performed to obtain the frequency domain symbols. Then the received signal can be given by
\begin{equation}
\begin{aligned}
        \mathbf{Y}_{k,m}=&\mathbf{Q}(\mathbf{H}_{k,m}(\mathbf{r},\mathbf{t})\mathbf{X} +\mathbf{N}_{k,m}),\\
\end{aligned}
\end{equation}
where $\mathbf{H}_{k,m}$ denotes the channel matrix in the $m$-th time slot at the $k$-th subcarrier in frequency domain, $\mathbf{Q}$ represents the combiner matrix which is elaborated by collecting all the combiner vectors, and the additive white Gaussian noise on the $k$-th subcarrier at the $m$-th time slot is denoted by $\mathbf{N}_{k,m}$.
\subsection{Time Varying Channel Framework}
To capture the variation of the CSI over time,  as illustrated in Fig. \ref{pilot}, the entire transmission period is partitioned into $N$ successive time blocks. For MA MIMO, the antennas can move in the local regions, indicating that the positions of MA can be changed over each time blocks for optimizing the communication performance. Each block is further divided into finer time slots.  {The first $M$ time slots are selected for channel estimation and each time  {slot}  {contains} $N_s$ OFDM symbols. The limited symbol period ensures minimal channel variation in several time slots \cite{11069254}. The selection of time blocks $N$ depends on prior knowledge, such as the rate of change of channel parameters and the adjustment speed of the MA system, to ensure that channel estimation and antenna position optimization can be accomplished within a single time block. When the CSI varies over time, dynamic channel estimation can be achieved by adjusting the density of time blocks and time slots, which determines the degree of sampling of the channel along the time dimension.} Given the short duration of each  {slot}, the channel can be considered approximately static within several time slots in a time block.
\begin{figure}[H]
    \centering    \includegraphics[width=\linewidth]{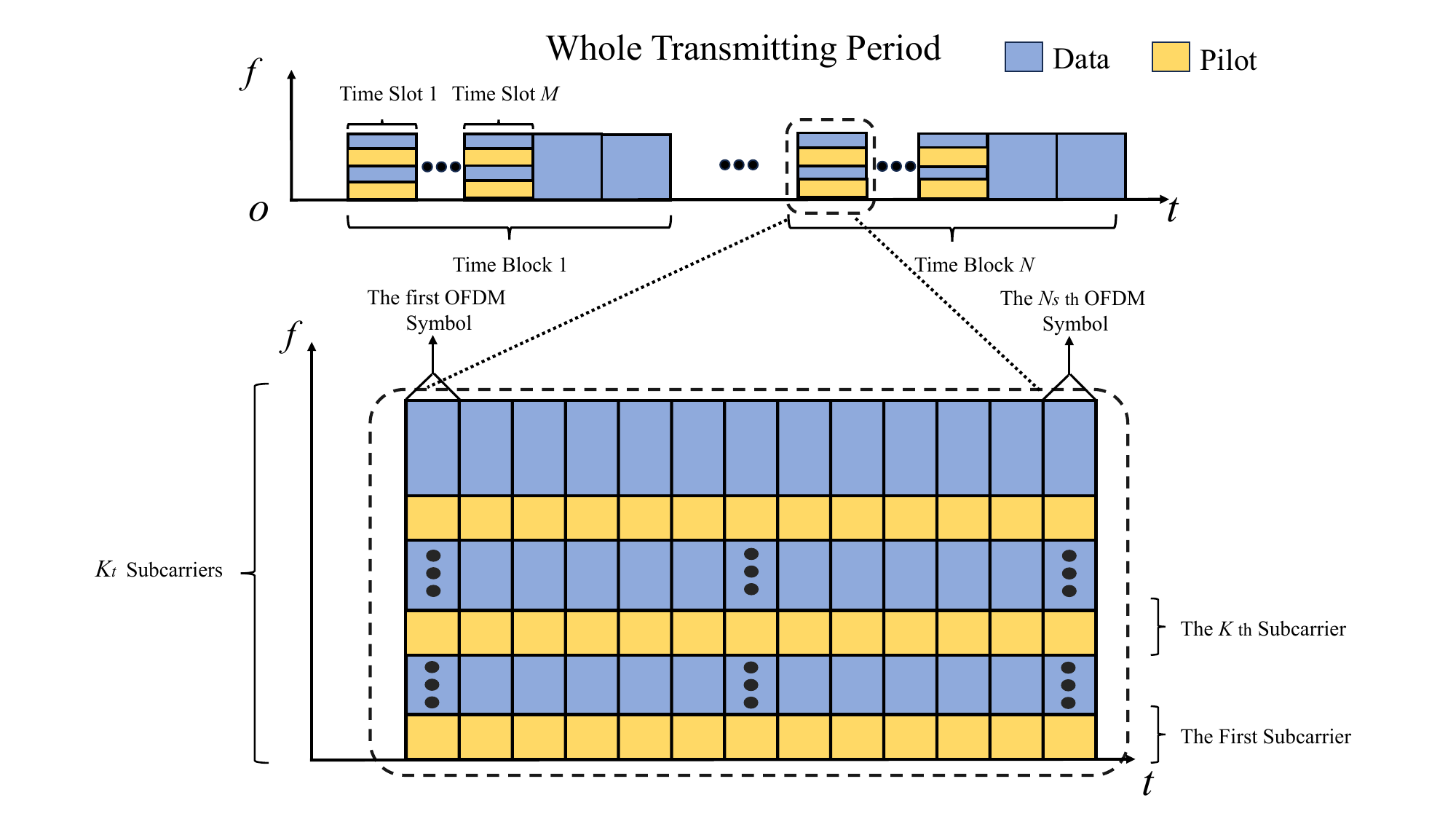}
    \caption{Pilot transmission structure for dynamic channel estimation.}
    \label{pilot}
\end{figure}
By leveraging the intrinsic sparsity in mmWave channels, the channel corresponding to MA pairs at BS and MS can be considered as a superposition of channels from sparse paths in the spatial domain. We consider a multi-path channel model with $R$ propagation paths. Thus, the time domain channel matrix can be written as\cite{10318061} 
\begin{equation}
\mathbf{H}_{t,\tau}(\mathbf{r}_n,\mathbf{t}_n) = \sum_{r = 1}^{R} \mathbf{f}(\mathbf{r}_n;\theta_r)( \beta_{r} e^{j2\pi \nu_{r}t}\delta(\tau - \tau_r)) \mathbf{g}(\mathbf{t}_n;\phi_r) ^T,\label{shiyuxindaojuzhen}
\end{equation}
where $\beta_r$, $\nu_r$, and $\tau_r$ denote the propagation loss following the complex Gaussian distribution, frequency offset caused by mobility, and time delay in the $r$-th path, respectively. The discrete form in time domain of (\ref{shiyuxindaojuzhen}) can be written as
 {\begin{equation}
    \begin{aligned}
    &\mathbf{H}(m,\tau)=\sum_{r = 1}^{R}\mathbf{f}(\mathbf{r}_n;\theta_r) \boldsymbol{\Sigma}_{\tau,m[r,r]}\mathbf{g}(\mathbf{t}_n;\phi_r)^T\\
    &=\!\!\!\sum_{r = 1}^{R}\!\mathbf{f}(\mathbf{r}_n;\theta_r)\!(\beta_{r} e^{j2\pi \nu_r(m-1)N_sT_s} \delta(\tau - \tau_r)\!)\mathbf{g}(\mathbf{t}_n;\phi_r)^T\!\!\!,
    \end{aligned}
    \label{88}
\end{equation}}
where $T_s = 1/\Delta f$ denotes the symbol duration in OFDM systems.  $\boldsymbol{\Sigma}_{m,\tau}$ denotes the time domain path response matrix between the local regions of MA in the BS and MS. The diagonal elements of $\boldsymbol{\Sigma}_{m,\tau}$ can be written as
\begin{equation}
    \boldsymbol{\Sigma}_{m,\tau [r,r]} = \beta_{r} e^{j2\pi \nu_{r}(m-1)N_sT_s} \delta(\tau - \tau_r).
\end{equation}

Based on (\ref{88}), the channel matrix in the $m$-th time slot at the $k$-th subcarrier in frequency domain can be given by
\begin {equation}
\begin{aligned}
    &\mathbf{H}(k,m)= \mathbf{F}\boldsymbol{\Sigma}_{k,m}  \mathbf{G}^T\\
    &= \!\!\sum_{r = 1}^{R}\!\mathbf{f}(\mathbf{r}_n;\theta_r)\!(\beta_{r}e^{j2\pi \nu_{r}(m-1)N_sT_s-j2\pi \frac{k}{K_t}f_s\tau_r}) \mathbf{g}(\mathbf{t}_n;\phi_r)^T
\!\!\!,\label{pinyuxindaojuzhen}
\end{aligned}
\end {equation}
where $f_s$ denotes sampling frequency. Existing channel estimation schemes for MA systems rely on continuously adjusting antenna positions to acquire channel information for all candidate locations, leading to considerable pilot overhead and time waste due to antenna movement. Notably, owing to the compact size of mmWave antennas, the field response channels across all MA position pairs share a common finite set of parameters, which naturally suggests adopting tensor representations for efficient channel estimation \cite{10659325}.

\subsection{Tensor Modeling for Received Signals}
We define that
 {\begin{align}
    \mathbf{b}_r&=\Big[1,e^{-j2\pi(f_s\Delta_p\tau_r)/K_t},\cdots,e^{-j2\pi (K-1)(f_s\Delta_p\tau_r)/K_t}\Big]^T,\\
    \mathbf{c}_r&=\Big[1,e^{j2\pi \nu_rN_sT_s},\ \cdots\ ,\ e^{j2\pi \nu_r(M-1)N_sT_s}\Big]^T,
     \label{cdyinzixiangliang}
\end{align}}
are the factor vectors and we have $\boldsymbol{\Sigma}_{m,k[r,r]} = \mathbf{b}_r[k]^T\mathbf{c}_r[m]$. Aggregating the received signals from the $K$ pilot subcarriers and $M$ time slots yields a fourth-order CP decomposition tensor representation, which can be written as $\mathcal{Y}$, i.e.,
\begin{align}
        \mathcal{Y}&=\!\bigg(\!\sum_{r=1}^R \beta_r\mathbf{f}(\mathbf{r}_n;\theta_r)\circ\mathbf{b}_r\circ\mathbf{c}_r\circ\mathbf{X}^T\mathbf{g}(\mathbf{t}_n;\phi_r)+\mathcal{N}\!\bigg)\times_1\!\mathbf{Q}\\
        &=\sum_{r=1}^R\!\bigg(\beta_r\mathbf{Q}^T\mathbf{f}(\mathbf{r}_n;\theta_r))\circ\mathbf{b}_r\circ\mathbf{c}_r\circ(\mathbf{X}^T\mathbf{g}(\mathbf{t}_n;\phi_r)\bigg)\!+\mathcal{N}^\mathbf{Q},\nonumber
        \label{tensor}
\end{align}
where $\times_n$ denotes the mode-$n$ product of tensor and $\mathcal{N}$ represents the noise tensor, whose entries follow i.i.d. circularly symmetric complex Gaussian distribution $\mathcal{CN}(0,\sigma_n^2)$, and $\mathcal{N}^{\mathbf{Q}}=\mathcal{N}\ \times_1 \mathbf{Q}$. Similar to (\ref{cdyinzixiangliang}), we define that $\mathbf{a}_r=\mathbf{Q}^T\mathbf{f}(\mathbf{r}_n;\theta_r)$ and $\mathbf{d}_r=\mathbf{X}^T\mathbf{g}(\mathbf{t}_n;\phi_r)$ are also the factor vectors of tensor $\mathcal{Y}$. As seen from (\ref{88}), the unknown parameters are nonlinearly coupled, which poses a challenge for joint estimation. By CP decomposition, this problem can be effectively solved. The factor  {matrices} of the tensor signals can be given by
\begin{equation*}
        \begin{aligned}
    \mathbf{A}=[\mathbf{Q}^T\mathbf{f}(\mathbf{r};\theta_1),\cdots,\mathbf{Q}^T\mathbf{f}(\mathbf{r};\theta_R)],\ \ \mathbf{B}=[\mathbf{b}_1,\mathbf{b}_2,\cdots,\mathbf{b}_R],\\
     \mathbf{C}=[\mathbf{c}_1,\mathbf{c}_2,\cdots,\mathbf{c}_R],
     \ \ \mathbf{D}=[\mathbf{X}^T\mathbf{g}(\mathbf{t};\phi_1),\cdots,\mathbf{X}^T\mathbf{f}(\mathbf{t};\phi_R)].
     \label{yinzijuzhen}
\end{aligned}
\end{equation*}
These four matrices $\{ \mathbf{A}, \mathbf{B}, \mathbf{C}, \mathbf{D}\}$ are factor matrices associated with a noiseless version of $\mathcal{Y}$, and the information of parameters are contained by these matrices.  {The received tensor signals $\mathcal{Y}$ can be rewritten as Tucker decomposition form}
{ 
\begin{equation}
\begin{aligned}
&\mathcal{Y}
=
\sum_{r=1}^{R}
\beta_r\,
\mathbf a_r\circ\mathbf b_r\circ
\mathbf c_r\circ\mathbf d_r
+\mathcal{N}^{\mathbf Q}\\
&=
\mathcal{D}_{\boldsymbol\beta}
\times_1\mathbf A
\times_2\mathbf B
\times_3\mathbf C
\times_4\mathbf D
+\mathcal{N}^{\mathbf Q}=
\mathcal Z+\mathcal{N}^{\mathbf Q},\\
\label{parafac}
\end{aligned}
\end{equation}
which explicitly captures the mathematical relationship between the received tensor signal and its corresponding factor matrices, where $\mathcal{Z}$ denotes the noise-free tensor signal and 
\begin{equation}
    [\mathcal{D}_{\boldsymbol\beta}]_{i_1,i_2,i_3,i_4}
=
\begin{cases}
\beta_r, & i_1=i_2=i_3=i_4=r,r=1,\ldots,R,\\
0, & \text{otherwise}.
\end{cases}
\end{equation}}
 {It should also be pointed out that (\ref{parafac}) illustrates the high-dimensional structure of the received tensor. For such characteristics, we can adopt high-order singular value decomposition (HOSVD) \cite{hosvd} to obtain a coarse estimate of the factor matrices, which provides an initialization approach for the iterative algorithm.}
\section{Tensor Decomposition-Based Channel Estimation Algorithm}
This section outlines a two-stage algorithm for channel parameter estimation utilizing tensor decomposition. In the first stage, we delineate the procedure for estimating the factor matrices  (\ref{yinzijuzhen}) through two distinct approaches: alternating least squares (ALS) and structured canonical polyadic decomposition (SCPD). Subsequently, the second stage describes the method for extracting the unknown parameters from the obtained factor matrices. Furthermore, the uniqueness condition of tensor decomposition, derivation of CRB and analysis of computational complexity are discussed. 

Owing to the inherent sparsity of mmWave channels, the number of propagation paths $R$ is typically limited compared to the signal dimensions. This sparsity results in an inherently low-rank composition of the tensor $\mathcal{Y}$ and provides favorable conditions for proving the uniqueness of the factor matrices abstraction\cite{10938032}.  Consequently, the parameter set $\left\{\beta_r, \nu_r, \theta_r, \phi_r, \tau_r\right\}$ can be effectively estimated  {through} CP decomposition of the received signal tensor $\mathcal{Y}$. The problem of factor matrices extraction for tensor signals can be formulated as
\begin{equation}
\min_{\{\hat\beta_r,\hat{\mathbf{a}}_r,\hat{\mathbf{b}}_r,\hat{\mathbf{c}}_r,\hat{\mathbf{d}}_r\}_{r=1}^{R}}
\Big\|\mathcal{Y}-\sum_{r=1}^R \hat\beta_r\hat{\mathbf{a}}_r\circ\hat{\mathbf{b}}_r\circ\hat{\mathbf{c}}_r\circ\hat{\mathbf{d}}_r\Big\|^2_F .
\label{fenjiewenti}
\end{equation}
\subsection{CP Decomposition Method}

\subsubsection{ALS}
To capture the high-dimensional structural information inherent in tensor received signals, an ALS based tensor decomposition algorithm is presented. This algorithm iteratively solves for the factor matrices until convergence criteria are met, at which point the solution is considered an optimal estimation of these matrices. This optimization problem can be efficiently solved through alternating minimization, which iteratively reduces the data fitting error by updating one factor matrix while keeping the remaining fixed \cite{10938032,7914672}. The computational procedure is outlined as follows:
\begin {align}
    \hat{\mathbf{A}}^{(t+1)}&=\text{arg}\min_{\hat {\mathbf{A}}}\Big\|\mathbf{Y}_{(1)}\!-\!\hat{\mathbf{A}}(\mathbf{D}^{(t)}\odot \mathbf{C}^{(t)} \odot\mathbf{B}^{(t)})^T\Big\|_F^2,    \label{alstendec}\\
    \hat{\mathbf{B}}^{(t+1)}&=\text{arg}\min_{\hat {\mathbf{B}}}\Big\|\mathbf{Y}_{(2)}-\hat{\mathbf{B}}(\mathbf{D}^{(t)}\odot \mathbf{C}^{(t)} \odot\mathbf{A}^{(t+1)})^T\Big\|_F^2,\nonumber\\
    \hat{\mathbf{C}}^{(t+1)}&=\text{arg}\min_{\hat {\mathbf{C}}}\Big\|\mathbf{Y}_{(3)}\!-\hat{\mathbf{C}}(\mathbf{D}^{(t)}\odot \mathbf{B}^{(t+1)} \odot\mathbf{A}^{(t+1)})^T\!\Big\|_F^2,\nonumber\\
    \hat{\mathbf{D}}^{(t+1)}&=\!\text{arg}\min_{\hat {\mathbf{D}}}\!\Big\|\mathbf{Y}_{(4)}\!-\hat{\mathbf{D}}(\mathbf{C}^{(t+1)}\odot \mathbf{B}^{(t+1)} \odot\mathbf{A}^{(t+1)})^T\!\Big\|_F^2,\nonumber
\end {align}
where $\mathbf{Y}_{(n)}$ denotes the mode-$n$ unfolding of the tensor $\mathcal{Y}$.  {These optimization problems can be solved by the least squares method.}  {It should be pointed out that the ALS algorithm is an iteration-based algorithm and is therefore sensitive to initialization. The HOSVD algorithm can be adopted as an initialization strategy, which extracts the dominant signal subspace of each mode from the observed tensor, thereby making the initial factors closer to the true low-rank structure. In addition, the estimation results of other algorithms, like multiple signal classification (MUSIC)  can also be used as initialization conditions, and the ALS algorithm can then be employed to obtain more accurate estimation results. Compared with purely random initialization, they typically provide faster convergence and lower risk of local optima.} Note that the obtained factor matrices are linked to the true factor matrices by
\begin{equation}
    \begin{aligned} 
\widehat{\mathbf{A}} = \mathbf{A}{\mathbf{\Lambda}}_{1}\mathbf{\Pi} + {\mathbf{E}}_{1}, \quad\quad\widehat{\mathbf{B}} = \mathbf{B}{\mathbf{\Lambda}}_{2}\mathbf{\Pi} + {\mathbf{E}}_{2}, \\ \widehat{\mathbf{C}} = \mathbf{C}{\mathbf{\Lambda}}_{3}\mathbf{\Pi} + {\mathbf{E}}_{3}, \quad\quad\widehat{\mathbf{D}} = \mathbf{D}{\mathbf{\Lambda}}_{4}\mathbf{\Pi} + {\mathbf{E}}_{4} ,
\end{aligned}
\end{equation}
where $\{\mathbf{\Lambda}_{1}, \mathbf{\Lambda}_{2}, \mathbf{\Lambda}_{3}, \mathbf{\Lambda}_{4}\}$ are unknown scaling diagonal matrices which satisfy $\mathbf{\Lambda}_{1}\mathbf{\Lambda}_{2}\mathbf{\Lambda}_{3}\mathbf{\Lambda}_{4} = \mathbf{I}$, and   $\mathbf{\Pi}$ denotes the unknown permutation characteristic of the recovered matrices. $\mathbf{E}_{i}, \ i\in\{1,\cdots,4\}$ denotes the error corresponding to the four estimated factor matrices. It is noteworthy that the inherent scaling ambiguities of CP decomposition do not impede the estimation of angles and delays. This is because these ambiguities still maintain the correlation characteristics of the factor matrix columns.  {Owing to the scaling ambiguity, which overwhelms the influence of the complex gain $\beta$ on the signal, all factor matrices must be recovered first before the propagation loss parameter can be estimated.}Moreover, the permutation matrix $\boldsymbol{\Pi}$ can be disregarded in the next analysis since it is shared across all four factor matrices, thus not affecting the parameter estimation procedure.  {The procedure of the ALS algorithm can be found in \textbf{Algorithm} \ref{algals}}.

\begin{algorithm}

\caption{ALS Algorithm for Tensor Decomposition}

\begin{algorithmic}[1]

{ 
\REQUIRE Received tensor
$\mathcal{Y}\in\mathbb{C}^{I_1\times I_2\times I_3\times I_4}$
and CP rank $R$.

\ENSURE Factor matrices
$\hat{\mathbf{A}}$,
$\hat{\mathbf{B}}$,
$\hat{\mathbf{C}}$, and
$\hat{\mathbf{D}}$.

\STATE Compute the mode-$n$ unfoldings
$\mathbf{Y}_{(n)}$,
$n=1,2,3,4$.

\STATE Initialize
$\hat{\mathbf{A}}$, $\hat{\mathbf{B}}$,
$\hat{\mathbf{C}}$, and $\hat{\mathbf{D}}$
using HOSVD or MUSIC.

\WHILE{$i<I_{\max}$ \textbf{and} $\epsilon>\epsilon_{\mathrm{tol}}$}

    \STATE Update the first factor matrix
    \[
    \hat{\mathbf{A}}=
    \mathbf{Y}_{(1)}
    \left[
    \left(\hat{\mathbf{D}}\odot
    \hat{\mathbf{C}}\odot
    \hat{\mathbf{B}}\right)^{T}
    \right]^{\dagger}.
    \]

    \STATE Update $\hat{\mathbf{B}}$, $\hat{\mathbf{C}}$, and
    $\hat{\mathbf{D}}$ similarly using the corresponding tensor unfoldings.

\ENDWHILE
}

\end{algorithmic}
\label{algals}
\end{algorithm}

ALS algorithm exhibits strong generality, making it applicable to any CP decomposition problem without imposing specific structural requirements on the factor matrices. Furthermore, it readily accommodates various constraints, as each least squares step can be easily adapted into a constrained least squares formulation, rendering it suitable for diverse and flexible application situations\cite{ 6166354}.
\subsubsection{SCPD}

While the ALS-based approach is theoretically straightforward and simple operation, it fails to incorporate the inherent structural constraints of factor matrices and tends to impose limitations on the maximum number of resolvable targets\cite{10403776}. To mitigate this problem, by leveraging the inherent structural information of the Vandermonde matrices, we introduce a method that directly computes the factor matrices through matrix algebra, eliminating the need for the iterative solving process in ALS algorithm. This approach avoids potential convergence issues associated with iterative algorithms, thereby enhancing the stability of the CP decomposition.  {Specifically, the SCPD algorithm consists of two steps. First, the received signal is smoothed along the time and frequency dimensions. Second, the factor matrices are recovered by exploiting the subspace of the smoothed matrix. Since the received signal is constructed as a fourth-order tensor, both signal smoothing and factor matrix decomposition are more complicated than in the third-order tensor case.}

 {First, spatial smoothing is applied to smooth the factor matrices of the received signal along the time and frequency dimensions. By  reshaping and transposing the first and fourth two modes with second and third two modes of received signal, we can obtain $\mathbf{Y}_{AD,BC}\triangleq \left(\mathbf{D}\odot\mathbf{A}\right)\boldsymbol{\Lambda}_{\beta}
\left(\mathbf{C}\odot\mathbf{B}\right)^T$},  {where $\boldsymbol{\Lambda}_{\beta}\triangleq\operatorname{diag}(\boldsymbol\beta)$. Define $\mathbf{Z}\triangleq\mathbf{Y}_{AD,BC}^{T}$ and the rows of $\mathbf{Z}$ are indexed by the two-dimensional pair $(k,m)$, with the frequency index $k$ varying fastest. Let $L_K$ and $L_M$ denote the frequency and time window lengths, respectively, where $1<L_K<K$ and $1<L_M<M$. For a window starting at $(p,\ell)$, the selected block is
\begin{equation}
{ 
\mathbf{Z}_{p,\ell}=
\left(\mathbf{C}_{L_M}\odot\mathbf{B}_{L_K}\right)
\mathbf{\Phi}_{p,\ell}\boldsymbol{\Lambda}_{\beta}
\left(\mathbf{D}\odot\mathbf{A}\right)^T,
}
\label{eq:local-smoothing-block}
\end{equation}}
 {where $p=1,\ldots,K-L_K+1$, $\ell=1,\ldots,M-L_M+1$, and $\mathbf{B}_{L_K}$ denotes the first $L_K$ columns of matrix $\mathbf{B}$, $\mathbf{C}_{L_M}$ follows the similar definition. The diagonal shift matrix is}
\begin{equation}
{ 
\mathbf{\Phi}_{p,\ell}=\operatorname{diag}\!\left(
z_{\tau,1}^{p-1}z_{\nu,1}^{\ell-1},\ldots,
z_{\tau,R}^{p-1}z_{\nu,R}^{\ell-1}\right),
}
\label{eq:two-dimensional-poles}
\end{equation}
 { where $z_{\tau,r}=e^{-j2\pi\Delta f\Delta_p\tau_r}$ and $z_{\nu,r}=e^{j2\pi N_sT_s\nu_r}$.  All forward windows are concatenated as $$\mathbf{H}_{\rm sm}=
\left[\mathbf{Z}_{1,1},\ldots,\mathbf{Z}_{K-L_K+1,M-L_M+1}\right].$$
Consequently, overlapping windows create  virtual snapshots without moving antennas or constructing a large observation covariance matrix. New smoothed matrix can be constructed such that the Khatri-Rao products on its left and right sides are restored to full column rank, thereby enabling the SVD to extract the complete $R$-dimensional signal subspace.}

{Next, the Vandermonde generators can be extracted from the signal subspace of the received signal, thereby enabling the recovery of the factor matrices. Let $\mathbf{H}_{\rm sm}=\mathbf{U}\boldsymbol{\Sigma}\mathbf{V}^{H}$ and let $\mathbf{U}_s=\mathbf{U}_{[:,1:R]}$  {denote} the signal
subspace obtained from the SVD of $\mathbf{H}_{\rm sm}$. Define the frequency-shift selection
matrices as
\begin{align}
    \mathbf{J}_{\tau,1}
=\mathbf{I}_{L_M}\otimes
\left[
\mathbf{I}_{L_K-1},\mathbf{0}_{(L_K-1)\times 1}
\right],\\
\mathbf{J}_{\tau,2}
=\mathbf{I}_{L_M}\otimes
\left[
\mathbf{0}_{(L_K-1)\times 1},\mathbf{I}_{L_K-1}
\right].
\end{align}
Accordingly, the two frequency-shifted subspaces are $\mathbf{U}_{\tau,1}=\mathbf{J}_{\tau,1}\mathbf{U}_s$ and $\mathbf{U}_{\tau,2}=\mathbf{J}_{\tau,2}\mathbf{U}_s$,  {which denote} the selection matrices that delete the last and first frequency samples, respectively. Similarly, the time-shift selection matrices are defined as
\begin{align} 
\mathbf{J}_{\nu,1}=
\left[
\mathbf{I}_{L_M-1},\mathbf{0}_{(L_M-1)\times 1}
\right]
\otimes\mathbf{I}_{L_K},\\
 \mathbf{J}_{\nu,2}=
\left[
\mathbf{0}_{(L_M-1)\times 1},\mathbf{I}_{L_M-1}
\right]
\otimes\mathbf{I}_{L_K}.
\end{align}
Thus, the time-shifted subspaces are $\mathbf{U}_{\nu,1}=\mathbf{J}_{\nu,1}\mathbf{U}_s$ and $\mathbf{U}_{\nu,2}=\mathbf{J}_{\nu,2}\mathbf{U}_s$. Since $\mathbf{U}_s=(\mathbf{C}_{L_M}\odot\mathbf{B}_{L_K})\mathbf{T}$ for an unknown nonsingular matrix $\mathbf{T}$, one-sample shifts in the frequency and time directions provide two ESPRIT relations within this one SCPD subspace. Therefore, the two regularized ESPRIT rotation matrices are given by
\begin{align}
\mathbf{\Psi}_{\tau}
&=
\mathbf{U}_{\tau,1}^{\dagger}\mathbf{U}_{\tau,2}
=\mathbf{T}_0^{-1}\mathbf{Z}_{\tau}\mathbf{T}_0,
\\
\mathbf{\Psi}_{\nu}
&=
\mathbf{U}_{\nu,1}^{\dagger}\mathbf{U}_{\nu,2}
=\mathbf{T}_0^{-1}\mathbf{Z}_{\nu}\mathbf{T}_0.
\label{eq:two-dimensional-esprit}
\end{align}
Let $\mathbf{\Psi}_{\tau}
=\mathbf{V}\operatorname{diag}(\hat{\mathbf z}_{\tau})
\mathbf{V}^{-1}$ and the eigenvalues $\hat{\mathbf z}_{\tau}$ are the estimated delay
Vandermonde generators. Since $\mathbf{\Psi}_{\tau}$ and
$\mathbf{\Psi}_{\nu}$ share the same similarity transformation, the
paired Doppler generators are obtained by $\hat{\mathbf z}_{\nu}
=\operatorname{diag}\left(
\mathbf{V}^{-1}\mathbf{\Psi}_{\nu}\mathbf{V}
\right)$. Then, the generator pairs
$\{(\hat z_{\tau,r},\hat z_{\nu,r})\}_{r=1}^{R}$ are obtained without
independently sorting two eigenvalue sets. The corresponding
Vandermonde factor columns are reconstructed as
\begin{equation*}
\hat{\mathbf b}_r=
\left[1,\hat z_{\tau,r},\ldots,
\hat z_{\tau,r}^{K-1}\right]^T,
\hat{\mathbf c}_r=
\left[1,\hat z_{\nu,r},\ldots,
\hat z_{\nu,r}^{M-1}\right]^T.
\end{equation*}
After estimating the two Vandermonde generators, the delay and Doppler factor matrices $\hat{\mathbf B}$ and $\hat{\mathbf C}$ are constructed. Define ${ \hat{\mathbf H}_{BC}\triangleq \hat{\mathbf C}\odot\hat{\mathbf B}}$ and $\mathbf{G} = \mathbf D\odot\mathbf A\boldsymbol{\Lambda}_{\beta} $, we have $\hat{\mathbf G}=\mathbf{Y}_{AD,BC} \left(\hat{\mathbf H}_{BC}^{T}\right)^\dagger$. Hence, the $r$th column of $\mathbf G$ is
\begin{equation}
\mathbf g_r=\left[\mathbf G\right]_{:,r}=\beta_r\left(\mathbf d_r\otimes\mathbf a_r\right)=
\operatorname{vec}\!\left(\beta_r\mathbf a_r\mathbf d_r^{T}\right).
\label{eq:g-column}
\end{equation}}
 {We rearrange the $r$-th recovered column into a matrix form
\begin{equation}
\hat{\mathbf G}_r = \operatorname{unvec}\left(\hat{\mathbf g}_r\right) = \beta_r\mathbf a_r\mathbf d_r^{T}
\in\mathbb{C}^{Q_{\rm MS}\times N_s}.
\label{eq:rank-one-matrix}
\end{equation}
 {The matrix in \eqref{eq:rank-one-matrix} is rank one
in the noise-free case. By exploiting this property, a truncated singular value decomposition (SVD) is performed on $\hat{\mathbf G}_r$ as $\hat{\mathbf G}_r \approx \sigma_{r,1} \mathbf u_{r,1} \mathbf v_{r,1}^{H}$, thereby enabling the estimation of the remaining factor matrices.}
\begin{equation}
\tilde{\mathbf a}_r=\sqrt{\sigma_{r,1}}\mathbf u_{r,1},
\quad
\tilde{\mathbf d}_r=\sqrt{\sigma_{r,1}}\mathbf v_{r,1}^{*}.
\label{eq:rank-one-factor-separation}
\end{equation}
It should be noted that, owing to the scaling ambiguity, the complex gain parameter $\beta$ must be estimated after all the other parameters have been obtained.}

\begin{algorithm}
\caption{ {SCPD algorithm for tensor decomposition}}
\renewcommand{\algorithmicrequire}{\textbf{INPUT:}}
\renewcommand{\algorithmicensure}{\textbf{OUTPUT:}}
\begin{algorithmic}[1]
 
\REQUIRE Received tensor $\mathcal{Y}\in\mathbb{C}^{Q_{\rm MS}\times N_s\times K\times M}$, rank $R$, and window lengths $L_K,L_M$.
\ENSURE $\hat{\mathbf A},\hat{\mathbf B},\hat{\mathbf C},\hat{\mathbf D}$.
\STATE Form $\mathbf{Z}=\mathbf{Y}_{AD,BC}^{T}$ and construct $\mathbf{H}_{\rm sm}$ by concatenating all $L_K\times L_M$ forward frequency-time windows.
\STATE Compute $\mathbf{U}_s=\mathbf{U}(:,1:R)$ from the SVD of $\mathbf{H}_{\rm sm}$.
\STATE Build the frequency and time shift subspaces and calculate $\mathbf{\Psi}_{\tau}$ and $\mathbf{\Psi}_{\nu}$ by regularized least squares.
\STATE Obtain paired $\hat z_{\tau,r}$ and $\hat z_{\nu,r}$ by eigendecomposing $\mathbf{\Psi}_{\tau}$ and jointly diagonalizing $\mathbf{\Psi}_{\nu}$.
\STATE Construct $\hat{\mathbf B}$ and $\hat{\mathbf C}$, recover $\widehat{(\mathbf{D}\odot\mathbf{A})\boldsymbol{\Lambda}_{\beta}}$, and recover the directions of $\hat{\mathbf A}$ and $\hat{\mathbf D}$ using \eqref{eq:rank-one-factor-separation}.
\end{algorithmic}
\label{scpdweidaima}
\end{algorithm}

\subsubsection{Analysis of Uniqueness}
 {
The CP decomposition is identifiable only up to a common column permutation and reciprocal column scalings. A standard sufficient condition for the fourth order model in (\ref{parafac}) is 
\begin{equation}
k(\mathbf A)+k(\mathbf B)+k(\mathbf C)+k(\mathbf D)\geq 2R+3.
\label{eq:cp-uniqueness}
\end{equation}
This is one of the most widely used uniqueness conditions. By reasonably designing the structures of matrices $\mathbf{S}$ and $\mathbf{Q}$, or by adjusting the numbers of time slots and subcarriers, the uniqueness of tensor decomposition can be fully satisfied. In particular, since the antenna positions can be changed, $\mathbf{A}$ and $\mathbf{D}$ can be adjusted to make it easier to satisfy this condition. However, it should also be noted that, because the factor matrices possess a Vandermonde structure, further exploiting this structure yields a more relaxed uniqueness condition.}

 {
\begin{lemma}
For a Vandermonde matrix with distinct generators and a generic second factor matrix, the Khatri-Rao product is generically full rank up to its dimensional limit, i.e.,
\[
r\left(\mathbf A^{(1)}\odot\mathbf A^{(2)}\right)
=
\min(I_1I_2,R).
\]
    \label{van}
\end{lemma}
\begin{proof}
Let $d=\min(I_1I_2,R)$. Since $\mathbf A^{(1)}\odot\mathbf A^{(2)} \in\mathbb C^{I_1I_2\times R}$, its rank is upper bounded by $d$. Assume that $\mathbf A^{(1)}$ is a Vandermonde matrix with distinct generators. To show that the upper bound is generically attainable, partition the first $d$ columns into at most $I_2$ disjoint groups $\{\mathcal I_j\}_{j=1}^{I_2}$, where $|\mathcal I_j|\leq I_1$. Construct one particular matrix $\overline{\mathbf A}^{(2)}$ by setting $\overline{\mathbf a}^{(2)}_r=\mathbf e_j$ for every $r\in\mathcal I_j$, where $\mathbf e_j$ is the $j$th canonical basis vector of $\mathbb C^{I_2}$. Up to a row permutation, the first $d$ columns of the resulting Khatri-Rao product become
\begin{equation}
\mathbf P\left(\mathbf A^{(1)}\odot\overline{\mathbf A}^{(2)}\right)_{[:,1:d]}=\operatorname{blkdiag}\left(\mathbf A^{(1)}_{[:,\mathcal I_1]},\ldots,\mathbf A^{(1)}_{[:,\mathcal I_{I_2}]}\right),
\end{equation}
where $\operatorname{blkdiag}$ denotes the construction of a block diagonal matrix. Because the generators of $\mathbf A^{(1)}$ are distinct and $|\mathcal I_j|\leq I_1$, every block $\mathbf A^{(1)}_{[:,\mathcal I_j]}$ has full column rank. Hence,
\begin{equation}
r\left(\mathbf A^{(1)}\odot\overline{\mathbf A}^{(2)}\right)\geq\sum_{j=1}^{I_2}|\mathcal I_j|=d.
\end{equation}
Therefore, at least one $d\times d$ minor of $\mathbf A^{(1)}\odot\mathbf A^{(2)}$ is a nonzero polynomial in the entries of $\mathbf A^{(2)}$. Its zero set has measure zero. Consequently, for generic $\mathbf A^{(2)}$,
\begin{equation}
r
\left(
\mathbf A^{(1)}\odot\mathbf A^{(2)}
\right)
=d
=
\min(I_1I_2,R).
\end{equation}
\end{proof}}
 {
\begin{proposition}
Let $\mathcal{X}\in\mathbb{C}^{I_1\times I_2\times I_3\times I_4}$
admit the CP representation
\begin{equation}
\mathbf{X}^{[4]}
=
(\mathbf{A}\odot\mathbf{B}\odot\mathbf{C})\mathbf{D}^{T},
\end{equation}
where $\mathbf{B}$ and $\mathbf{C}$ are Vandermonde matrices with
distinct generators. By applying two-dimensional smoothing along the
second and third modes, consider
\begin{equation}
\mathbf{X}_{\rm sm}
=
\left(
\mathbf{A}\odot\mathbf{B}_{(k_1)}
\odot\mathbf{C}_{(k_2)}
\right)
\left(
\underline{\mathbf{B}}_{(l_1)}
\odot\mathbf{C}_{(l_2)}
\odot\mathbf{D}
\right)^T,
\label{eq:smoothed-matrix-ABCD}
\end{equation}
where $k_1+l_1=I_2+1$ and $k_2+l_2=I_3+1$.
If
\begin{equation}
\begin{cases}
r\!\left(
\underline{\mathbf{B}}_{l_1}
\odot\mathbf{C}_{l_2}
\odot\mathbf{D}
\right)=R,\\[1mm]
r\!\left(
\mathbf{A}\odot\mathbf{B}_{k_1}
\odot\mathbf{C}_{k_2}
\right)=R,
\end{cases}
\label{eq:scpd-full-rank-ABCD}
\end{equation}
then $\operatorname{rank}(\mathcal{X})=R$, and the CP decomposition of
$\mathcal{X}$ is essentially unique under the Vandermonde constraints.
\end{proposition}
\begin{proof}
Define $\mathbf{F}=\mathbf{A}\odot\mathbf{B}_{k_1}\odot\mathbf{C}_{k_2},\quad \mathbf{G}=\underline{\mathbf{B}}_{l_1}\odot\mathbf{C}_{l_2}\odot\mathbf{D}$. From \eqref{eq:scpd-full-rank-ABCD}, both $\mathbf{F}$ and
$\mathbf{G}$ have full column rank $R$. Consequently, $r(\mathbf{X}_{\rm sm})=R$. An additional one-sample smoothing along the second mode yields an
auxiliary third-order CP model whose factors can be expressed as 
\begin{equation}
    \mathbf{X}_{\rm ss}=(\mathbf{F}\odot\mathbf{E})\mathbf{G}^{T}, 
\label{eq:auxiliary-cpd-ABCD}
\end{equation}
where $\mathbf{E}=\mathbf{B}_{2}$. Since $\mathbf{B}$ is Vandermonde with distinct generators, any two-row submatrix $\mathbf{E}=\mathbf{B}^{(2)}$ has Kruskal rank $k(\mathbf{E})=2$. Moreover, \textbf{Lemma} \ref{van}  {implies} $k_{\mathbf{F}}=R$ and $k_{\mathbf{G}}=R$. Therefore,
\begin{equation}
k_{\mathbf{F}}+k_{\mathbf{E}}+k_{\mathbf{G}}
=
R+2+R
=
2R+2,
\end{equation}
which satisfies Kruskal's uniqueness condition. Hence, the auxiliary
CP decomposition in \eqref{eq:auxiliary-cpd-ABCD} is essentially unique. In particular, the $r$-th generator of $\mathbf{B}$ is obtained from the corresponding column of $\mathbf{E}$. It should be noted that, since the factor matrix $\mathbf{C}$ also exhibits a Vandermonde structure, the factor matrix $\mathbf{C}$ can likewise be obtained through similar steps. The remaining factors $\mathbf{A}$, and $\mathbf{D}$ can be recovered from $\mathbf{F}$, $\mathbf{G}$, and the original tensor by rank-one factor separation and least-squares fitting. Thus, the decomposition of $\mathcal{X}$ is unique up to the inherent common column permutation and scaling ambiguities.
\end{proof}}

In our considered model under generic case, condition (\ref{eq:scpd-full-rank-ABCD}) can be rewritten as
\begin{equation}
    \text{min}\left((l_1-1)l_2N_s, k_1k_2Q_\text{MS}\right) \geq R.
\end{equation}
Given the typically small number of paths $R$, it is plausible to assume that the product of RF chains $Q_\text{MS}$ and $k_1k_2$ exceeds $R$, particularly when $k_1, k_2, l_1, l_2 > 2$. Correspondingly, the number of OFDM symbols $N_s$ within a time slot generally surpasses the path count $R$, especially when multiplied by the factors $(l_1 - 1)$ and $l_2$. Thus, the condition $(l_1 - 1) l_2 N_s \geq R$ generally holds in practical scenarios. In summary, for the hybrid precoding matrix $\mathbf{F}_\text{RF}\mathbf{F}_{k,m}$, training symbol matrix $\mathbf{X}$, and combining matrix $\mathbf{Q}$ with entries uniformly selected from a unit circle, the proposed method requires only $N_s = R/k_1k_2$ and $Q_\text{MS} = R/(l_1 - 1)l_2$ to satisfy the uniqueness condition (\ref{eq:scpd-full-rank-ABCD}). In practical implementations, slightly larger values of $N_s$ and $Q_\text{MS}$ may be necessary to achieve accurate channel estimation due to the presence of observation noise and estimation errors.  {It should also be noted that, since a fourth-order tensor is constructed to characterize the time-varying channel, an additional factor matrix $\mathbf{C}$ with a Vandermonde structure is introduced, thereby allowing a more relaxed uniqueness condition \eqref{eq:scpd-full-rank-ABCD} to be obtained. More importantly, by adjusting the antenna positions, the MA system can modify the structure of the factor matrix along the spatial dimension, thereby making the uniqueness condition more readily satisfied. This also reflects the relationship between the MA system and tensor decomposition algorithms.}

\subsection{Channel Parameters Extraction Method}
Following the recovery of the factor matrices, we process the estimation of the channel parameters. It should be noted that the $r$-th column of $\mathbf{A}$ and $\mathbf{B}$ are intrinsically linked to the angular parameters in the $r$-th path. Hence, the parameter $\hat{\theta_r}$ and $\hat{\phi_r}$ can therefore be estimated by using a correlation based approach, as detailed in the following expression
\begin{align}
    \hat{\theta_r} &= \mathop{\arg\max}\limits_{\hat{\theta_r}} \frac{ \Big \lvert \mathbf{a}_r^H \mathbf{a}(\hat{\theta_r}) \Big \rvert}{\big\lVert \mathbf{a}_r\big\rVert_2  \big\lVert \mathbf{a}(\hat{\theta_r})\big\rVert_2},\label{thetaest}\\
    \hat{\phi_r} &= \mathop{\arg\max}\limits_{\hat{\phi_r}} \frac{ \Big \lvert \mathbf{d}_r^H \mathbf{d}(\hat{\phi_r}) \Big \rvert}{\big\lVert \mathbf{d}_r\big\rVert_2  \big\lVert \mathbf{d}(\hat{\phi_r})\big\rVert_2}\label{phiest}.
\end{align}
 {Problems \eqref{thetaest} and \eqref{phiest} can be solved via grid search.} {It should also be noted that the algorithm proposed in this paper can be extended to  two-dimensional planar antenna arrays, in which case the azimuth and elevation angles involved can be estimated by two-dimensional search using the maximum correlation criterion. Alternatively, super-resolution estimation can be achieved by changing the antenna positions to satisfy a specific structure\cite{11087029}.} Moreover, note that the factor matrix $\mathbf{B}$ can be rewritten as product of the Vandermonde matrix $\mathbf{B}_{(1)}$ and a diagonal matrix $\mathbf{B}_{(2)}$, where
\begin{align}
    \mathbf{B}_{(1)} &= [\mathbf{b}_{(1)}^1,\cdots, \mathbf{b}_{(1)}^R ],\\
    \mathbf{b}_{(1)}^r &= [e^{-j2\pi(f_s\tau_r)/K_t},\cdots,e^{-j2\pi K(f_s\tau_r)/K_t}]^T,\\
    \mathbf{B}_{(2)} &= \text{diag}([\beta_1,\cdots,\beta_R]).
\end{align}
Thus, we have $\mathbf{B} = \mathbf{B}_{(1)}\mathbf{B}_{(2)}$ and the estimation of time delay $\hat{\tau_r}$ can be estimated by
\begin{equation}
        \hat{\tau_r}=-\frac{K_t}{2\pi{f}_s}\angle\mathbf{B}_{(2)[r,r]}=-\frac{K_t}{2\pi{f}_s}\angle({({\hat{\underline{\mathbf{b}}}_r})^\dagger}\overline{\hat{\mathbf{b}}}_r),\label{tauest}
\end{equation}
where $\underline{\mathbf{b}_r}$ and $\overline{\mathbf{b}}_r$ are obtained by deleting the top and bottom row of $\mathbf{b}$, respectively, and $\angle$ denotes the phase angles of a complex number. If the SCPD algorithm is adopted for tensor decomposition, the Vandermonde generator $z_{c,r}$ and $z_{c,r}$ obtained from the decomposition can be directly utilized to achieve parameter estimation. The parameter $\hat{\nu}_r$ and $\hat{\tau}_r$ can also be directly obtained following a procedure analogous to (\ref{tauest}), which can be given by
\begin{equation}
    \hat{\nu}_r = \frac{\angle z_{c,r}}{2\pi N_sTs}, \quad\hat{\tau}_r = -\frac{K_t\angle z_{b,r}}{2\pi{f}_s}.
    \label{nutau}
\end{equation}
This method not only avoids the estimation inaccuracy inherent in grid-based search techniques but also eliminates the requirement for calculating the pseudo-inverse of $\mathbf{b}_r$. Finally, the estimation of the propagation loss $\hat{\beta_r}$ can be estimated by
\begin{align}
    \hat{\beta_{r}}=\label{betaest}\mathop{\arg\min}\limits_{\beta_r} \bigg\lVert \!\!\left(\mathbf{Y}_{(2)}\!\left(\!(\mathbf{D} \odot \mathbf{C}\odot \mathbf{A})^{T}\!\right)^{\dagger}\!\right)_{[:,r]}\!\!\!-\!\!\beta_r\mathbf{b}\left(\hat{\tau_r}\right) \!\!\bigg\rVert^2_2.
\end{align}
To this end, the unknown parameters of the targets can be obtained at the MS and once these channel parameters are obtained, the channel matrix can be reconstructed via (\ref{pinyuxindaojuzhen}).

\subsection{ {CRB-Guided MA Position Optimization}}
In this subsection, we derive CRB for the channel parameter estimation problem considered in (\ref{88}). Details of the derivation can be found in Appendix \ref{ccb}. As is well known, the CRB is a lower bound on the variance of any unbiased estimator, which illustrates the behavior of the resulting bounds. For ease of exposition,  {let  $\mathbf{p}_r=\big[\theta_r,\phi_r,\tau_r,\nu_r,\operatorname{Re}\{\beta_r\},\operatorname{Im}\{\beta_r\}\big]^T$ and $\mathbf{p}=[\mathbf{p}_1^T,\ldots,\mathbf{p}_R^T]^T\in\mathbb R^{6R}$ denote the estimation parameter vector.} Thus, the log-likelihood function
of $\mathbf{p}$ can be expressed as
\begin{equation*}
    \begin{aligned}
        L(\mathbf{p}) &= Q_\text{MS}N_sKM - \frac{1}{\sigma_n^2}\big\|\mathbf{Y}_{(n)}-\mathbf{Z}_{(n)}\Big\|_F^2,
\end{aligned}
\end{equation*}
where $\mathbf{Z}_{(n)}$ denotes the mode-$n$ unfolding of $\mathcal{Z}$. The calculation of complex Fisher information matrix (FIM)  for $\mathbf{p}$ is defined as $\mathbf{\Omega}(\mathbf{p}) =\mathbb{E}\Big[\left({\partial L(\mathbf{p})}/{\partial \mathbf{p}}\right)\left({\partial L(\mathbf{p})}/{\partial\mathbf{p}}\right)^H\Big]$. Further details for the derivations of the log-likelihood function, the partial derivative of $L(\mathbf{p})$ with respect to $\mathbf{p}$ and the FIM can be found in Appendix \ref{ccb} and the CRB can be derived as the inverse of the matrix $\mathbf{\Omega}(p)$. 

 {The CRB can serve as a metric for parameter estimation accuracy, and can further serve as the objective function for antenna position optimization. For example, guided by the estimated CSI, adjusting the antenna positions can minimize the sum of the normalized CRB of all parameters. For instance, the receive steering vector and its angular derivative depend explicitly on the selected MA coordinates. Thus, the FIM and CRB therefore  {provide} a natural bridge from channel estimation to antenna positioning. First, the antenna positions are optimized based on the estimated parameter  {information,} and a convenient design criterion is the normalized average-variance optimality (A-optimal) objective}
\begin{equation}
{ 
\mathcal P_{\rm MS}^{\star}=
\arg\min_{\substack{\mathcal P_{\rm MS}\subseteq\mathcal P_{\rm cand}\\
|\mathcal P_{\rm MS}|=N_{\rm MS}}}
\operatorname{tr}\!\left[\left(\widetilde{\mathbf{\Omega}}
(\hat{\mathbf{p}}^{(1)};\mathcal P_{\rm MS})\right)^{-1}\right],
}
\label{eq:aoptimal-ma-design}
\end{equation}
 {where $\widetilde{\mathbf{\Omega}}$ is the FIM after normalizing the heterogeneous angles, delay, Doppler, and propagation loss coordinates. There are many approaches to solving the optimization problem \eqref{eq:aoptimal-ma-design}, and it is not necessary to exhaustively enumerate all antenna position combinations. For instance, a greedy alternating one-location exchange strategy can be used to reduce the A-optimal objective. One selected antenna position is exchanged with an unselected candidate position whenever the exchange decreases the A-optimal criterion. The procedure terminates when no further improvement is found or when the maximum number of iterations is reached. The optimization problem is further subject to MA system constraints, such as the movement region and the minimum spacing. }

 {More importantly, in MA systems, channel parameter estimation can provide guidance for position optimization, and the optimized antenna positions can in turn improve the accuracy of parameter estimation\cite{10354003,TutMAnet}, which cannot be achieved in fixed-antenna systems.}

\subsection{Analysis of Complexity}

The whole computation procedure of our proposed algorithm can be divided into  {two stages: tensor decomposition} and parameter extraction from the factor matrices. For ALS, a single procedure in iteration (\ref{alstendec}) requires $\mathcal{O}\{Q_\text{MS}N_sKMR+(N_sKM+Q_\text{MS}KM+Q_\text{MS}N_sM+Q_\text{MS}N_sK)R^2+R^3\}$   multiply operations, which can be simplified to $\mathcal{O}\{Q_\text{MS}N_sKMR\}$ as the dominant term. Then the total computation complexity of ALS is $\mathcal{O}\{N_\text{iter}Q_\text{MS}N_sKMR\}$, where $N_\text{iter}$ denotes the required max iteration steps.  {For SCPD, let $k_1$ and $k_2$ denote the frequency- and time-domain smoothing-window lengths, and $l_1=K-k_1+1$ and $l_2=M-k_2+1$. The two-dimensional smoothing operation constructs a $k_1k_2\times Q_{\rm MS}N_sl_1l_2$ matrix, with complexity $\mathcal{O}\{k_1k_2Q_{\rm MS}N_sl_1l_2\}$. The subsequent SVD has complexity $\mathcal{O}\{(k_1k_2)^2Q_{\rm MS}N_sl_1l_2\}$, which is the dominant SCPD cost in the current implementation. }

\section{Numerical Simulations}
 {This section provides numerical simulations to assess the accuracy of the proposed algorithm. We also  {conduct a comparative analysis  with the baseline algorithms.} We consider a BS equipped with $N_\text{BS} = {\text{12}}$ transmit MA and a MS equipped with $N_\text{MS} = {\text{12}}$ receive MA. The number of RF chains at the BS and MS are $Q_\text{BS} = {\text{10}}$ and $Q_\text{MS} = {\text{6}}$. The carrier frequency is set to be $f_\text{c} = {\text{28}}$ GHz with the system bandwidth of $f_s = {\text{10}}$ MHz, which is divided into $K_t = {\text{128}}$ subcarriers. Thus, the space of  {each subcarrier} is $\Delta f = {\text{78.125}}$ kHz. $K = \text{16}$ subcarriers are selected  {as pilots} and are uniformly distributed across the operating bandwidth. Let $M = \text{8}$ time slots in each time  {block} and $N_s = \text{8}$ OFDM symbols in each time  {slot}. The SNR is formally expressed as the ratio of the average transmit power which is denoted by $P$, to the average noise power which is represented by $\sigma_n^2$, and can be given by $\text{SNR} = {P}/{\sigma_n^2}$. The complex channel gain $\beta$ for each path can be drawn from a circularly symmetric complex Gaussian distribution, the angular, frequency offset, and time delay  {parameters of the multipath components are generated from uniform distribution, and we suppose that there are} $R = \text{3}$ propagation paths.  In our simulations, all entries in $\mathbf{X}$ and $\mathbf{Q}$ are randomly generated. To satisfy the power constraint in the hybrid precoding architecture, the transmit power is normalized such that $\|\mathbf{X}_{[:,n]}\|^\text{2} = \text{1}/N_{\text{BS}}$ for $n = \text{1}, \text{2}, \cdots, N_s$.} 

\begin{figure*}[t]
    \centering
    \begin{tabular}{@{}ccc@{}}       
        \begin{subfigure}[b]{0.3\linewidth}
            \centering
            \includegraphics[width=\linewidth]{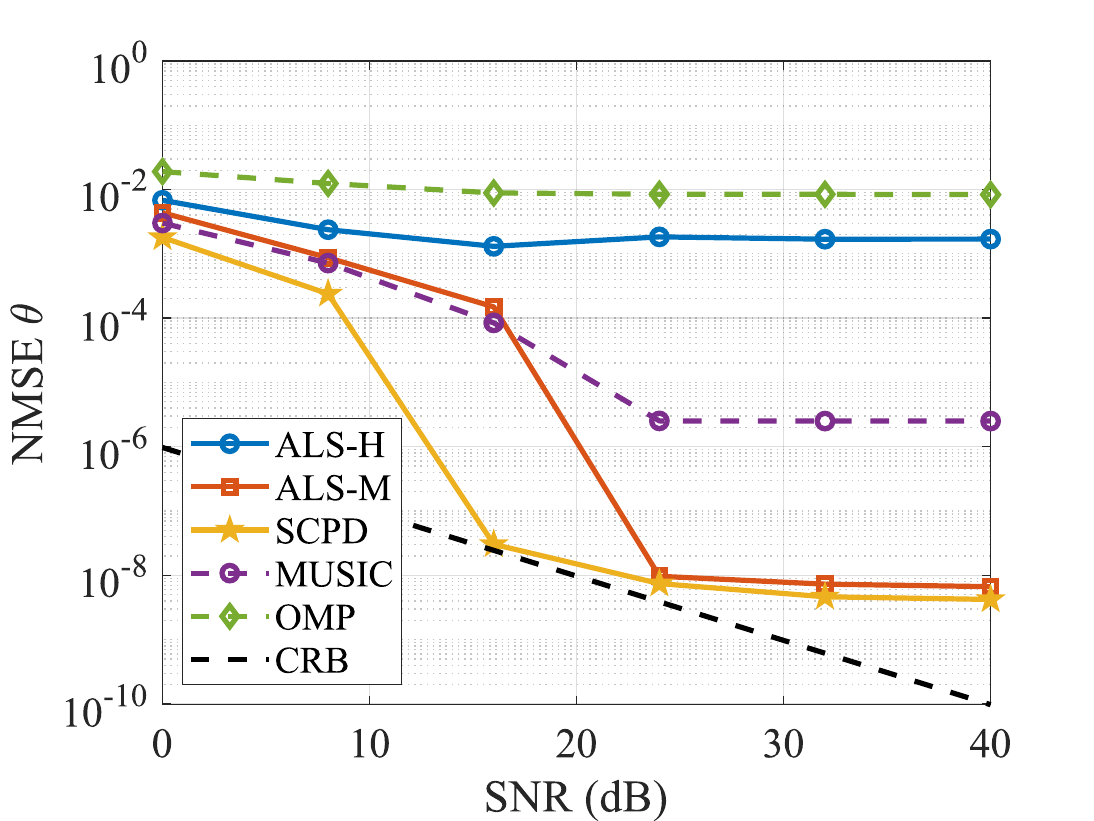}
            \caption{}
            \label{theta}
        \end{subfigure}
        &
        \begin{subfigure}[b]{0.3\linewidth}
            \centering
            \includegraphics[width=\linewidth]{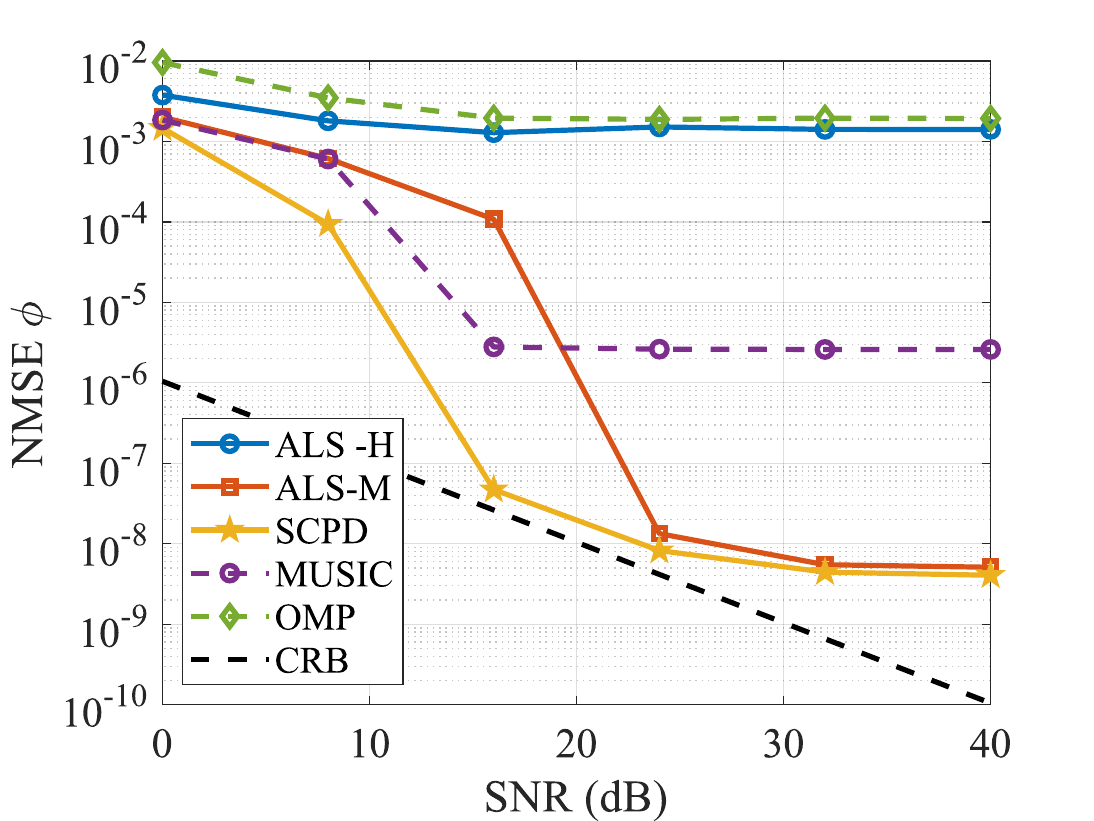}
            \caption{}
            \label{phi}
        \end{subfigure}    
        &
        \begin{subfigure}[b]{0.3\linewidth}
            \centering
            \includegraphics[width=\linewidth]{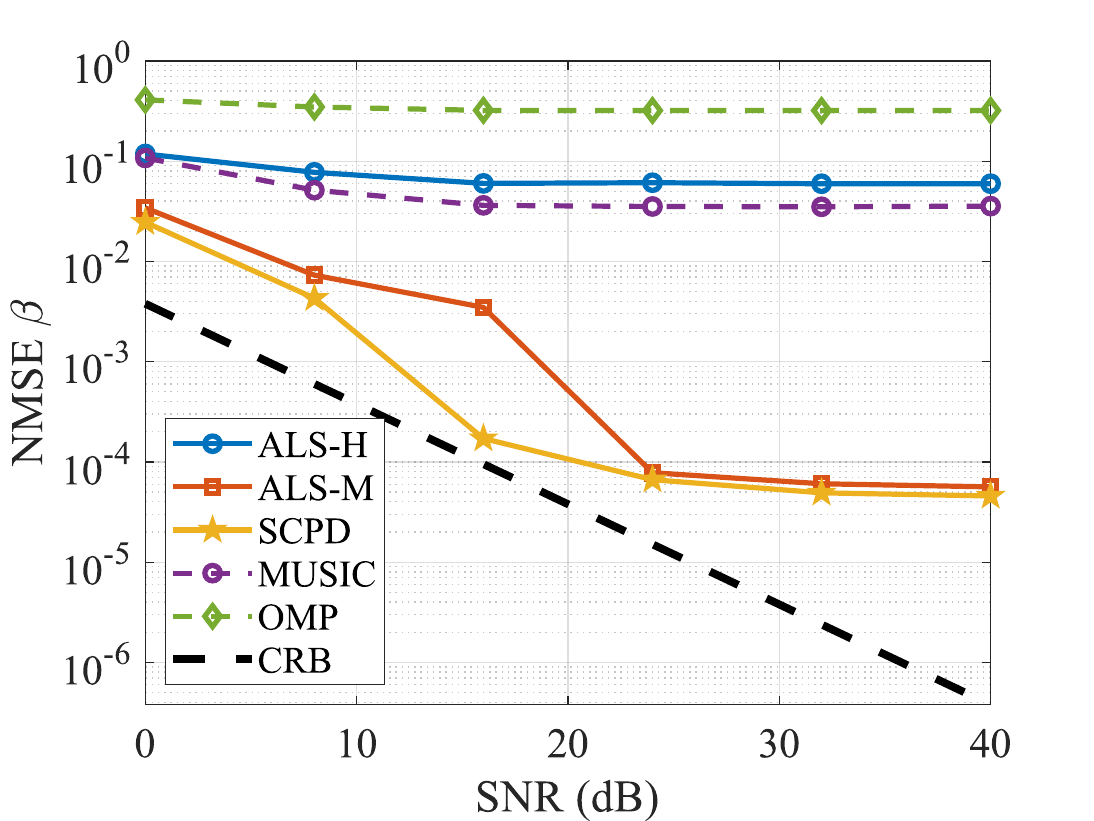}
            \caption{}
            \label{beta}
        \end{subfigure} \\
   
        \begin{subfigure}[b]{0.3\linewidth}
            \centering
            \includegraphics[width=\linewidth]{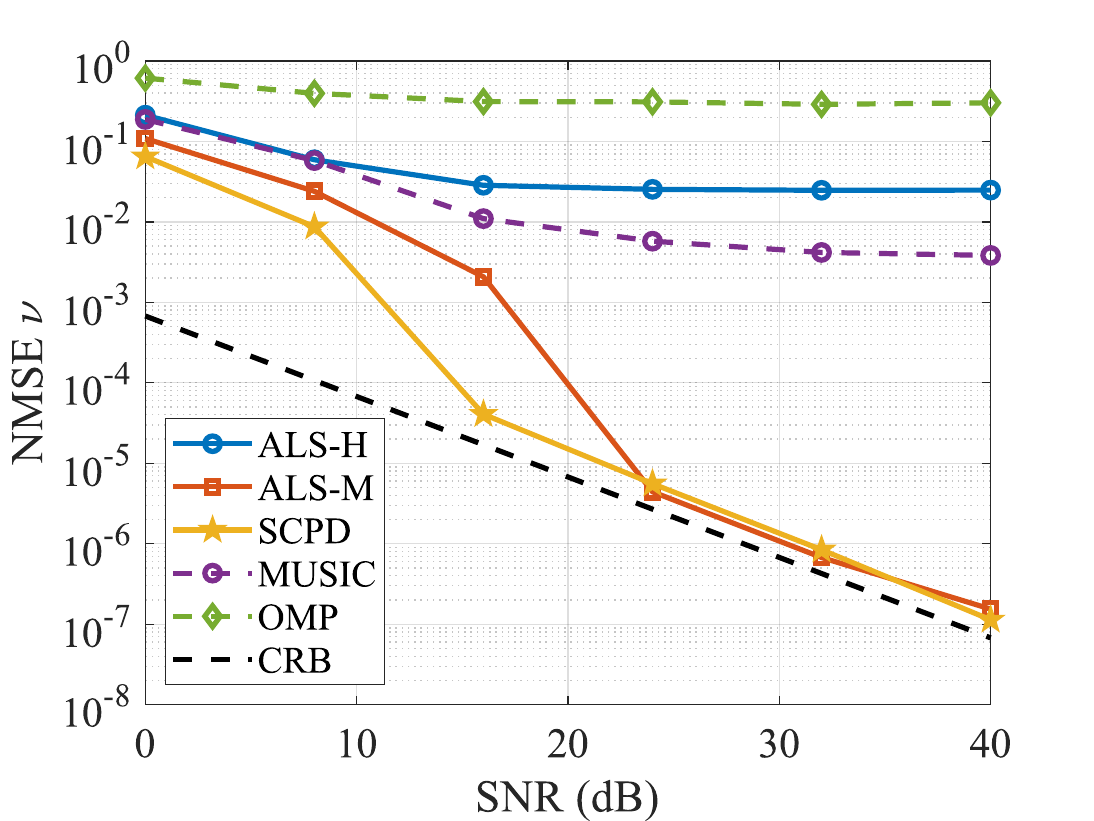}
            \caption{}
            \label{nu}
        \end{subfigure}
        &
        \begin{subfigure}[b]{0.3\linewidth}
            \centering
            \includegraphics[width=\linewidth]{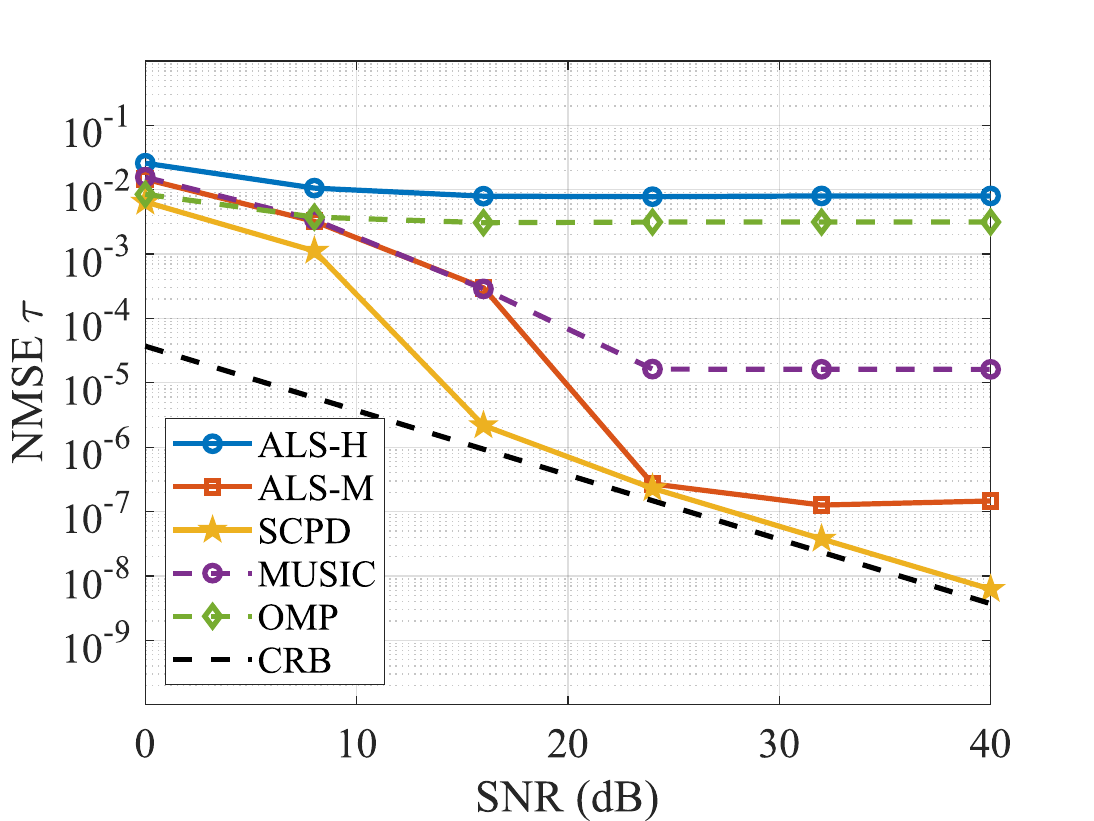}
            \caption{}
            \label{tau}
        \end{subfigure}  
        &
        \begin{subfigure}[b]{0.3\linewidth}
            \centering
            \includegraphics[width=\linewidth]{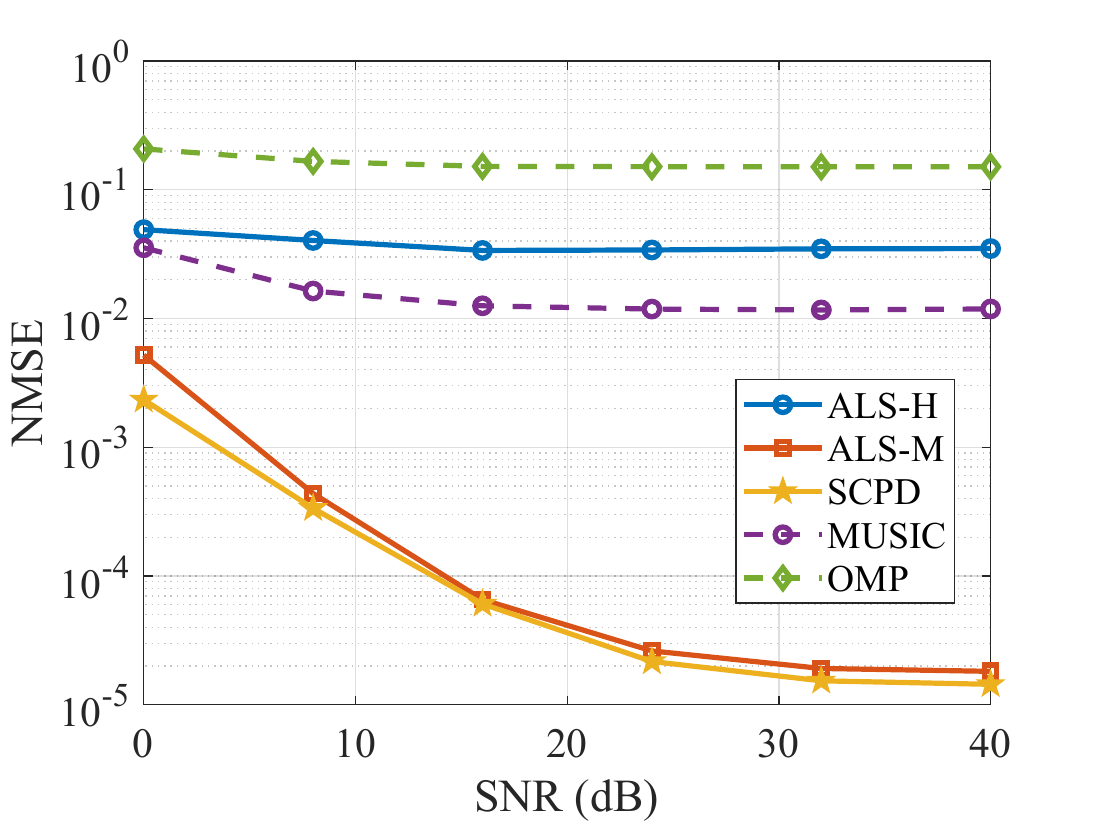}
            \caption{}
            \label{mseh}
        \end{subfigure}  
\end{tabular}
\caption{Parameter estimation and the rebuilt channel NMSE performance.}
\label{gujicanshu}
\end{figure*}
\subsection{Parameter Estimation Performance}
We first assess the estimation accuracy of the channel parameters. For each parameter set, normalized mean square error (NMSE) is selected as evaluation metric, computed as follows
\begin{equation*}
    \text{NMSE}(\mathbf{Z}) = \mathbb{E}\left\{\frac{\|\mathbf{Z}-\hat{\mathbf{Z}}\|^2}{\|\mathbf{Z}\|^2}\right\}, 
\end{equation*}
where $\mathbf{Z}\in\{\boldsymbol{\nu}, \mathbf{\Theta}, \mathbf{\Phi}, \boldsymbol{\tau}, \boldsymbol{\beta}\}$ and $\boldsymbol{\nu} = [\nu_1, \cdots, \nu_R]$, $\mathbf{\Theta} = [\theta_1, \cdots, \theta_R]$, $\mathbf{\Phi} = [\phi_1, \cdots, \phi_R]$, $\boldsymbol{\tau} = [\tau_1, \cdots, \tau_R]$ and  $\boldsymbol{\beta} = [\beta_1, \cdots, \beta_R]$.   Additionally, we introduce the minimum attainable variance by unbiased estimator CRB as a theoretical reference. Orthogonal matching pursuit (OMP)\cite{8813076} and MUSIC\cite{1143830} are applied to be the baseline algorithms in numerical simulations.  {The ALS algorithm adopts HOSVD (ALS-H) \cite{hosvd} and MUSIC (ALS-M) \cite{yueshuALS} as its initialization strategy.}

The estimation accuracy of both tensor decomposition-based algorithms improves as the SNR increases. In particular, the SCPD algorithm approaches the CRB when the SNR exceeds 15~dB. This is because the factor matrices containing the two parameters $\nu$ and $\tau$ exhibit a Vandermonde structure. By exploiting this property, the Vandermonde generators that carry the parameter information can be directly extracted from (\ref{nutau}), thereby avoiding the accuracy loss caused by grid-based estimation and achieving super-resolution estimation as shown in Fig. \ref{nu} and Fig. \ref{tau}.  {The accuracy loss induced by grid-based estimation is illustrated in Fig.~\ref{theta} and Fig.~\ref{phi}. Since the antenna positions in MA systems are not uniformly spaced, only the estimation algorithm in (\ref{thetaest}) and (\ref{phiest}) can be adopted, which is based on the correlation criterion. Such a grid-based search strategy is limited by the grid resolution. Consequently, when the SNR exceeds 20~dB, the estimation accuracy of these algorithms reaches a plateau, and the estimation error no longer decreases with increasing SNR.} However, the two tensor decomposition-based parameter estimation algorithms exploit the high-dimensional structural information of the signal and extract parameters from the factor matrices rather than directly from the received signal. Therefore, even in Fig.~\ref{theta} and Fig.~\ref{phi} where the estimation performance is limited by the grid resolution, the tensor decomposition-based algorithms still exhibit the best estimation performance.

 {Nevertheless, precisely because the parameters are extracted from the factor matrices, the errors incurred in the tensor decomposition step will propagate to the subsequent step, making the decomposition accuracy particularly critical. The ALS algorithm is sensitive to initialization conditions and iterative steps, and thus may exhibit poor estimation results at low SNR due to unstable iterations, as shown in Fig. \ref{phi}. The influence of different initialization conditions on the performance of the ALS algorithm, such as HOSVD and MUSIC, is also illustrated in Fig. \ref{gujicanshu}. In contrast, the SCPD algorithm further exploits the prior information of the Vandermonde structure in the factor matrices, thereby avoiding the convergence instability associated with iterative algorithms and thus achieving the best estimation accuracy. }

 {The MUSIC and OMP algorithms do not take into account the high-dimensional signal structure and simply flatten the high-dimensional received signal into a matrix for processing, which causes the information from different physical dimensions to be aliased together. Consequently, their estimation accuracy is inferior to that of the two tensor decomposition algorithms that exploit structural information. In particular, when the MUSIC estimation results are used as the initialization for the ALS algorithm, the error can be further reduced on the original basis, which also demonstrates that the two conventional algorithms are limited because they cannot exploit the structural information of high-dimensional signals. Moreover, both algorithms impose requirements on the orthogonality between the model column vectors, such as the restricted isometry property (RIP) in the OMP algorithm and the strict orthogonality requirement between the signal subspace and the noise subspace in the MUSIC algorithm. In the signal model, multiplication with the rank-deficient matrices $\mathbf{Q}$ and $\mathbf{X}$ reduces the observation dimensions and may also disrupt such orthogonality, thereby causing the two algorithms to fail. Therefore, these two conventional algorithms are also inferior to the tensor decomposition-based parameter estimation algorithms in terms of applicability. Finally, Fig.~\ref{beta} presents the estimation results of the path loss parameter $\beta$. The estimation of the path loss parameter is based on the estimation results of the preceding four parameters, and the estimation errors of these four parameters accumulate here. It can also be observed from the simulation results that the estimation accuracy of $\beta$ is the lowest among all parameters. Nevertheless, the proposed algorithm still exhibits the best estimation performance.}

\subsection{Channel Estimation Performance}

Then, we focus on the NMSE performance of the channel estimation methods. Based on above estimated channel parameters, the algorithm performance of the rebuilt channel is subsequently evaluated using NMSE, defined as
\begin{equation}
        \text{NMSE}_{\mathbf{H}} = \frac{\sum^M_{m=1}\sum^K_{k=1}\Big \lVert\mathbf{H}_{k,m} - \widehat{\mathbf{H}}_{k,m}\Big \rVert_F^2}{MK\sum^M_{m=1}\sum^K_{k=1}\Big \lVert \mathbf{H}_{k,m}\Big \rVert_F^2}.
\end{equation}

The superiority of the proposed algorithm is evident from the NMSE results for the rebuilt channel in Fig.~\ref{mseh}. As discussed, MUSIC and OMP impose stringent requirements. For parameter $\nu$, the observed dimension is insufficient to satisfy the RIP condition, making OMP ineffective. Moreover, their iterative atom selection leads to error accumulation, and both algorithms are highly noise-sensitive. They also require joint estimation of $\tau$ and $\nu$, resulting in poor performance in Fig.~\ref{tau} and degraded channel reconstruction accuracy. In contrast, tensor decomposition-based algorithms impose relaxed structural constraints and exhibit strong noise robustness. In particular, SCPD achieves the best accuracy in low SNR situations, as shown in Fig.~\ref{gujicanshu} and Fig.~\ref{mseh}.

Fig.~\ref{changeknmse} shows the NMSE versus the number of subcarriers at SNR = $\text{5}$ dB. The proposed algorithm clearly outperforms the benchmarks. For tensor decomposition-based algorithms, estimation accuracy improves as $K$ increases, whereas OMP shows no such enhancement. The proposed algorithm remains effective for angle, delay, and loss estimation, while for MUSIC and  OMP, $K$ mainly affects delay estimation alone, so its overall performance stays largely unchanged \cite{11069254}. Fig.~\ref{changeRF} illustrates the NMSE as the number of time slots increases.  {This indicates that as the number of time slots increases, each time block is divided more densely. In time-varying scenarios, this implies that the sampling along the time dimension becomes progressively finer; therefore, all the algorithms exhibit improved accuracy as the number of time slots increases.} The two tensor decomposition-based algorithms again demonstrate superior performance, similar to Fig.~\ref{changeknmse}.  The two sets of results  {in} Fig. \ref{change} demonstrate that as parameters such as the number of the length of signal increase, the tensor decomposition-based algorithm exhibits better estimation performance and more pronounced variations. This improvement stems from the fact that  the increase of row dimension of the factor matrices enriches the information in the corresponding dimensions, leading to better separation between the signal and noise.

\begin{figure}[H]
    \centering
    \begin{tabular}{@{}c@{}}
        \begin{subfigure}[b]{0.6\linewidth}
            \centering
            \includegraphics[width=\linewidth]{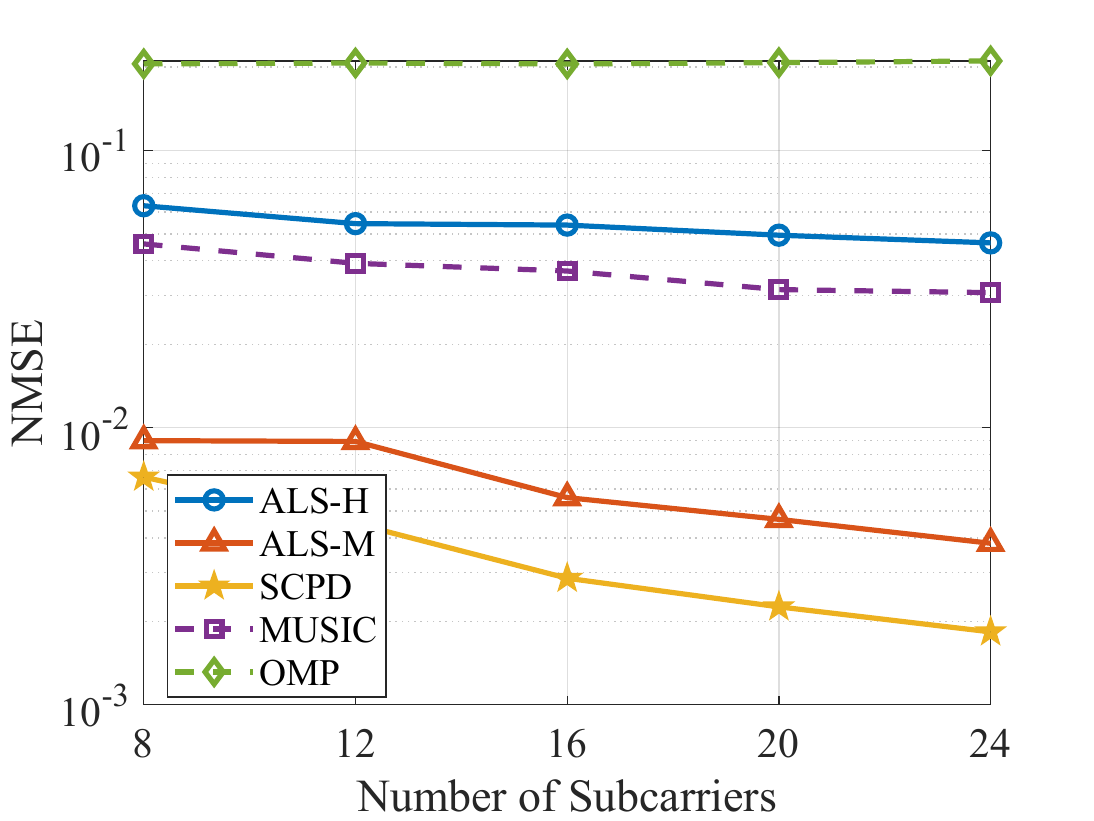}
            \caption{}
            \label{changeknmse}
        \end{subfigure}
        \\
        \begin{subfigure}[b]{0.6\linewidth}
            \centering
            \includegraphics[width=\linewidth]{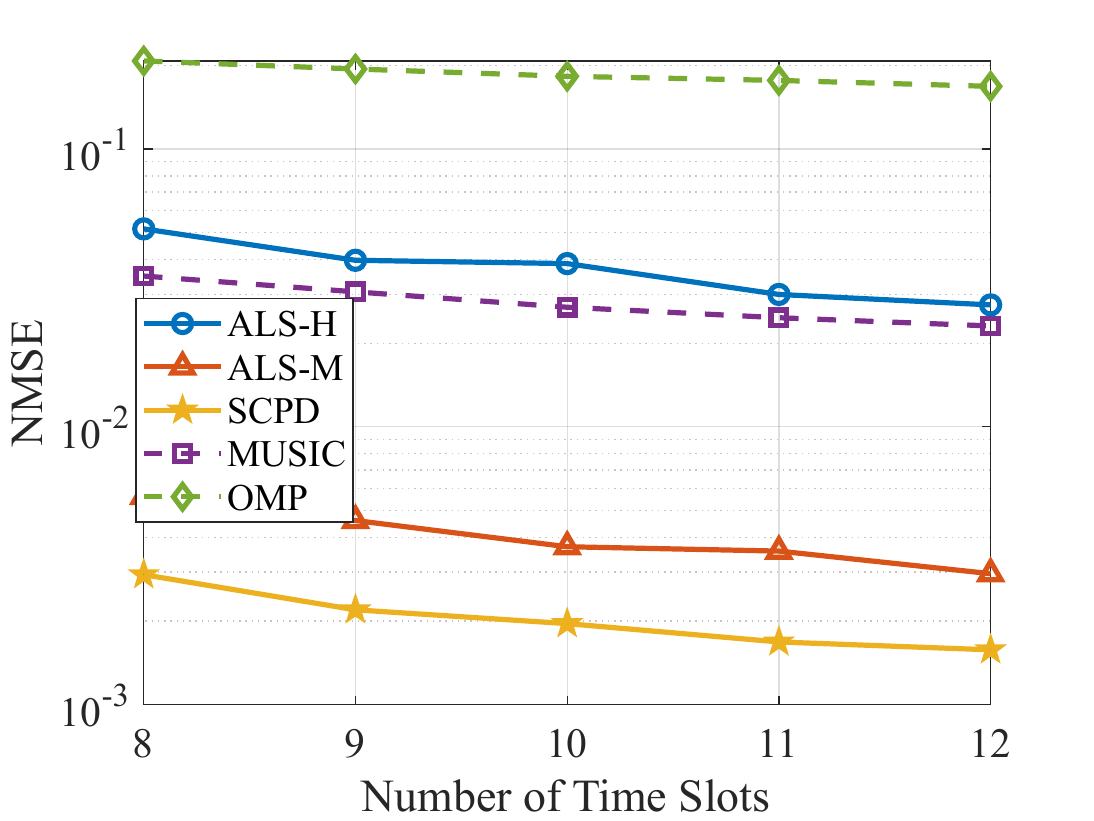}
            \caption{}
            \label{changeRF}
        \end{subfigure}   
\end{tabular}
\caption{NMSE of the rebuilt channel as the setting changes.}
\label{change}
\end{figure}

 {Fig.~\ref{diffR} presents the variation in the estimation performance of various algorithms under different numbers of paths. It can be observed that, as the number of paths increases, the estimation accuracy of all algorithms degrades to some extent. In particular, the MUSIC algorithm exhibits the largest variation, and correspondingly, the performance of the ALS algorithm initialized with the MUSIC results also varies relatively significantly. This is because, when the number of paths is large, spectral peaks in the MUSIC algorithm may overlap, leading to algorithm failure. For tensor decomposition algorithms, the number of paths corresponds to the CP rank, which directly affects the uniqueness condition. When the number of paths changes, the algorithm performance also varies. Nevertheless, under different conditions, the tensor decomposition-based parameter estimation algorithms, especially SCPD, still exhibit the best performance.}
\begin{figure}
    \centering
    \begin{tabular}{@{}cc@{}}       
        \begin{subfigure}[b]{0.5\linewidth}
            \centering
            \includegraphics[width=\linewidth]{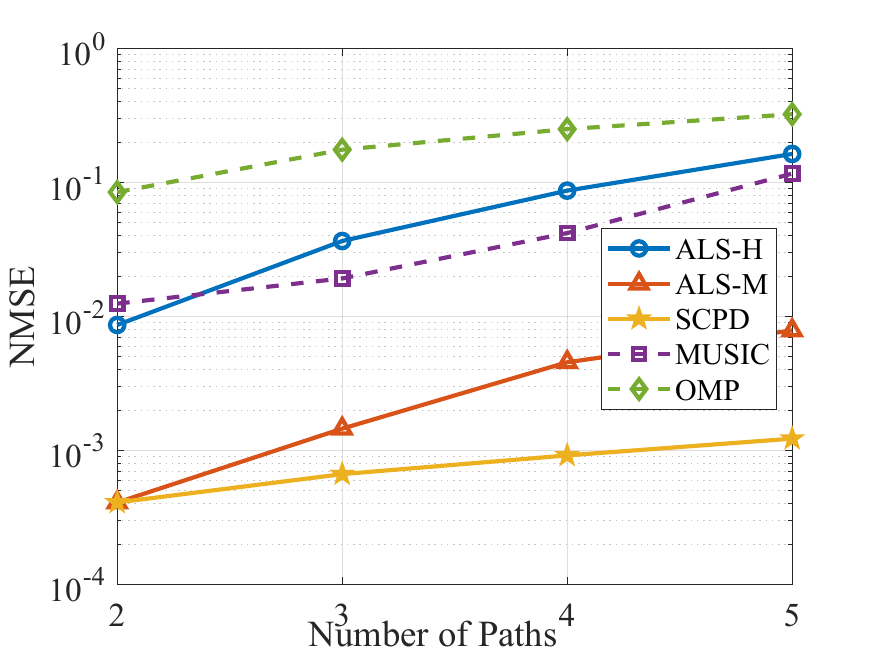}
            \caption{}
            \label{diffR}
        \end{subfigure}
        &
        \begin{subfigure}[b]{0.5\linewidth}
            \centering
            \includegraphics[width=\linewidth]{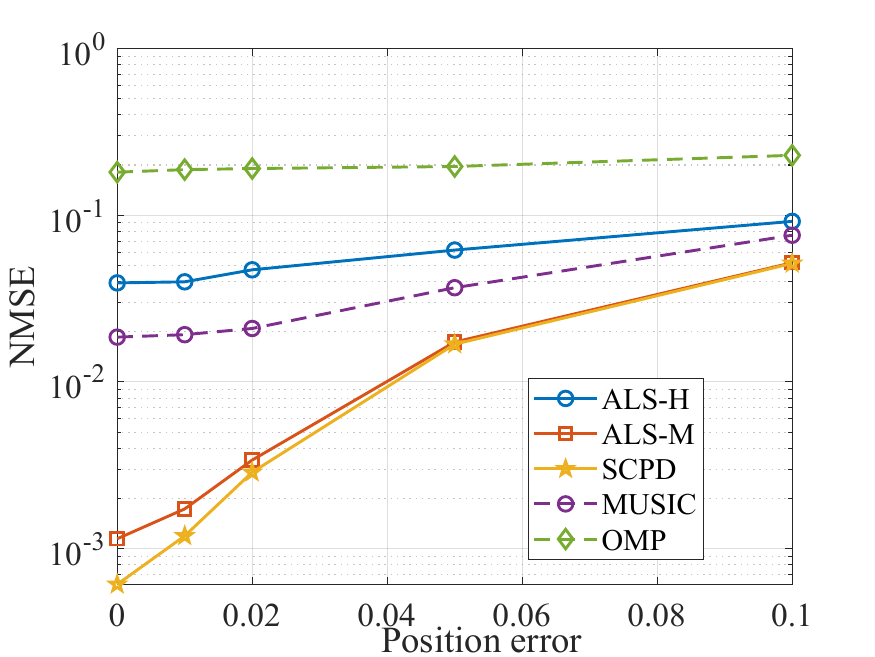}
            \caption{}
            \label{weizhiwucha}
        \end{subfigure}   
\end{tabular}
    \caption{Estimation performance of the algorithms under different number of paths and normalized antenna position error conditions.}
    
\end{figure}
 {Finally, we consider a practical scenario in which discrepancies $\sigma_p$ may exist between the actual and measured antenna positions during channel estimation. Therefore, it is necessary to investigate the robustness of the estimation algorithms under different normalized antenna position errors $\rho=\sigma_p/\lambda$. It can be seen from Fig.~\ref{weizhiwucha} that the proposed algorithm exhibits the best estimation performance, demonstrating robustness to antenna position errors.} 

\subsection{ {MA Position Optimization Based on Estimation Results}}
\begin{figure}[H]
    \centering
    \begin{tabular}{@{}cc@{}}       
        \begin{subfigure}[b]{0.5\linewidth}
            \centering
            \includegraphics[width=\linewidth]{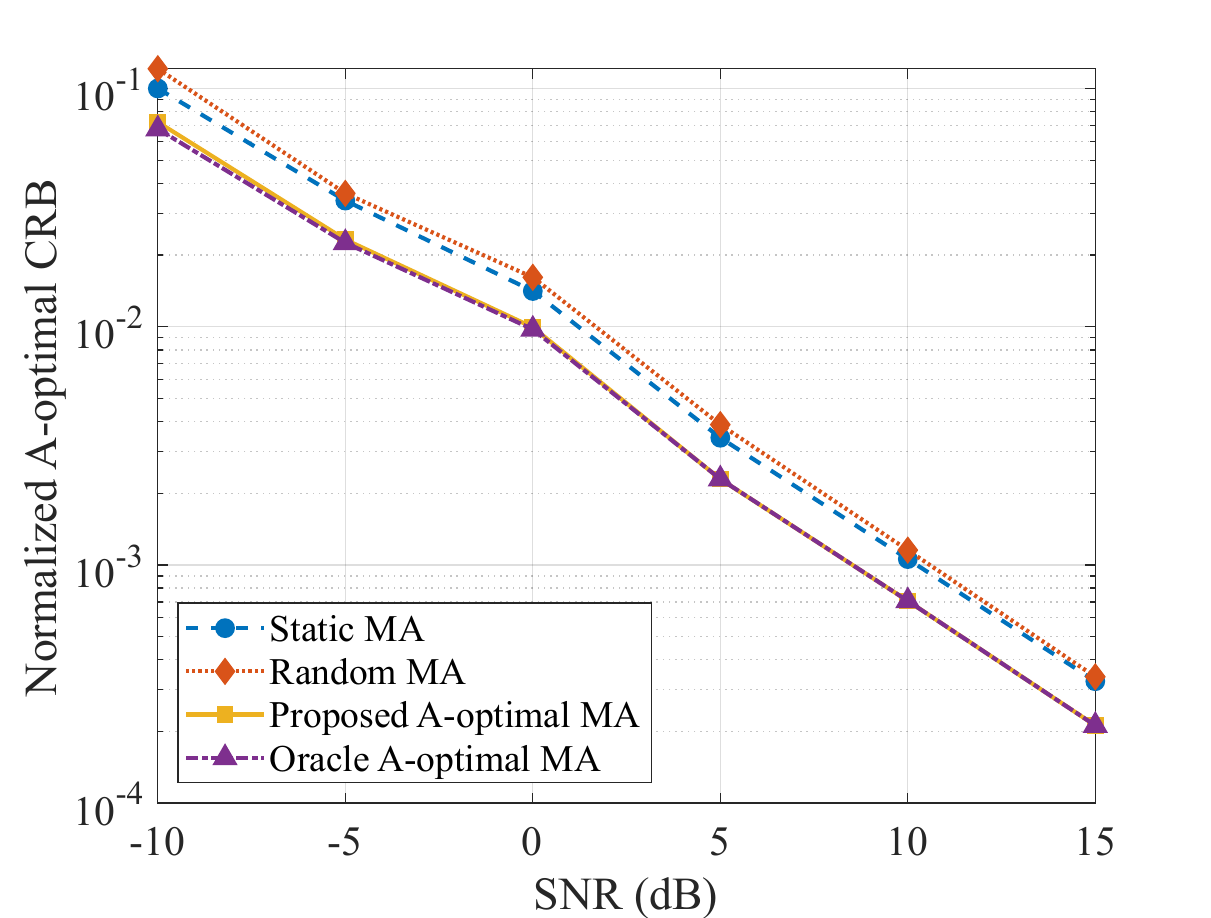}
            \caption{}
            \label{pcrb}
        \end{subfigure}
        &
        \begin{subfigure}[b]{0.5\linewidth}
            \centering
            \includegraphics[width=\linewidth]{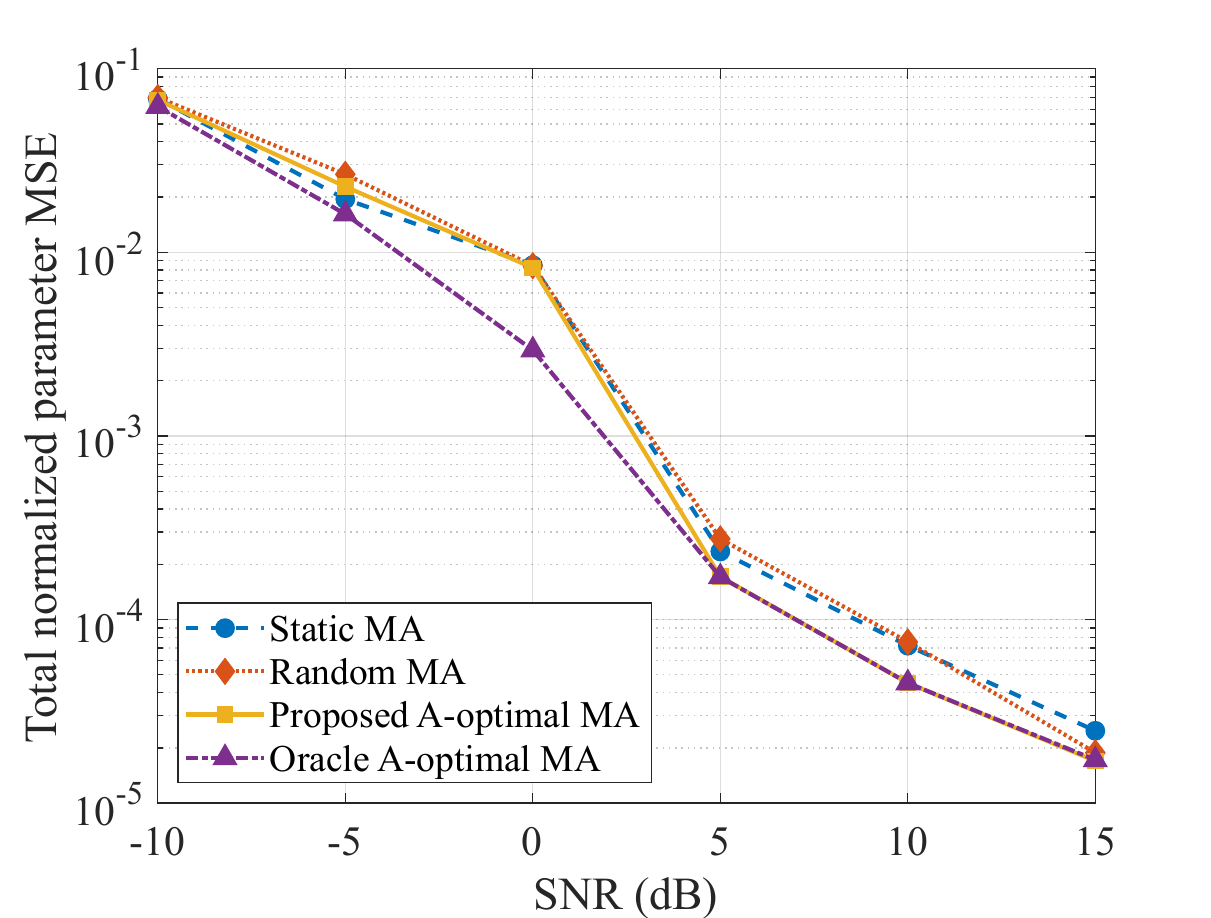}
            \caption{}
            \label{crbe}
        \end{subfigure}   
\end{tabular}
\caption{Performance of antenna position optimization: (a) A-optimal CRB achieved by antenna position optimization based on estimation results; (b) corresponding channel estimation performance after antenna position optimization.}
\label{ao}
\end{figure}

 {
Fig.~\ref{ao} demonstrates the benefit of channel estimation-guided MA position optimization. Specifically, the channel parameters estimated from an initial pilot block are used to configure the MA positions for the subsequent block, and the resulting optimized geometry is then further used for better channel parameter estimation. Fixed MA positions, randomly selected MA positions, the proposed estimation-guided design, and an oracle A-optimal design based on the true channel parameters are considered for comparison. In particular, Fig.~\ref{pcrb} presents the normalized A-optimal CRB achieved by the different antenna-position configurations. The position design minimizes the sum of the normalized CRBs of all channel parameters, equivalently the trace of the inverse normalized FIM. Using the parameters estimated by the proposed method, a greedy iterative search is adopted. As shown, the MA configuration designed from the estimated parameters attains a normalized CRB sum close to that of the oracle benchmark, and it significantly outperforms the fixed and random MA configurations. Fig.~\ref{crbe} further reports the corresponding channel parameter estimation error after position optimization. The estimation accuracy improves with the SNR, and the performance of the proposed estimation-guided design progressively approaches that of the oracle-aided configuration. This result confirms that the optimized MA geometry provides a more informative observation model and can further improve channel parameter estimation accuracy. Overall, accurate channel estimation enables effective MA position optimization, while the optimized antenna geometry improves the subsequent estimation performance, thereby establishing a beneficial closed loop between channel estimation and MA position design.}

\subsection{Runtime Performance}

{
\captionsetup[table]{}
\begin{table}
\centering
\caption{CPU Time of Different Algorithms}
\label{time}
 {\begin{tabular}{ccccc}
\toprule
\multirow{2}{*}{\textbf{Algorithms}} & \multicolumn{4}{c}{\textbf{CPU Time (s)}} \\
\cmidrule{2-5}
& Initialization & {Decomposition} & Estitmation & Total \\
\midrule
SCPD   & --    & 0.0202 & 0.0015 & 0.0217 \\
ALS-H    & 0.0008    & 0.0170 & 0.0013 & 0.0190 \\
ALS-M   & 0.0221   & 0.0164 & 0.0013 & 0.0401 \\
MUSIC  & --    & --     & 0.0221 & 0.0221 \\
OMP    & --    & --     & 0.0429 & 0.0429 \\
\bottomrule
\end{tabular}}
\end{table}
}

 {Table~\ref{time} lists the CPU time for one round of parameter estimation at SNR = \text{20} dB and $K=\text{12}$ subcarriers.} The tensor decomposition-based algorithms avoid grid search by extracting parameters from Vandermonde generators, making  {their estimation steps} faster than benchmark algorithms. OMP is the slowest due to iterative atom selection processing high-dimensional signals. Notably, the computational complexity of OMP, MUSIC and ALS scales linearly with the length of the received signal, whereas that of SCPD increases quadratically with the signal length. It can be observed that when the factor matrices of the tensor signal have small dimensions, the computational complexity of the SCPD algorithm is lower than all algorithms. In such cases, the complexity of ALS algorithm is largely dominated by number of iterations $N_\text{iter}$.
\begin{figure}[H]
    \centering
    \includegraphics[width=0.7\linewidth]{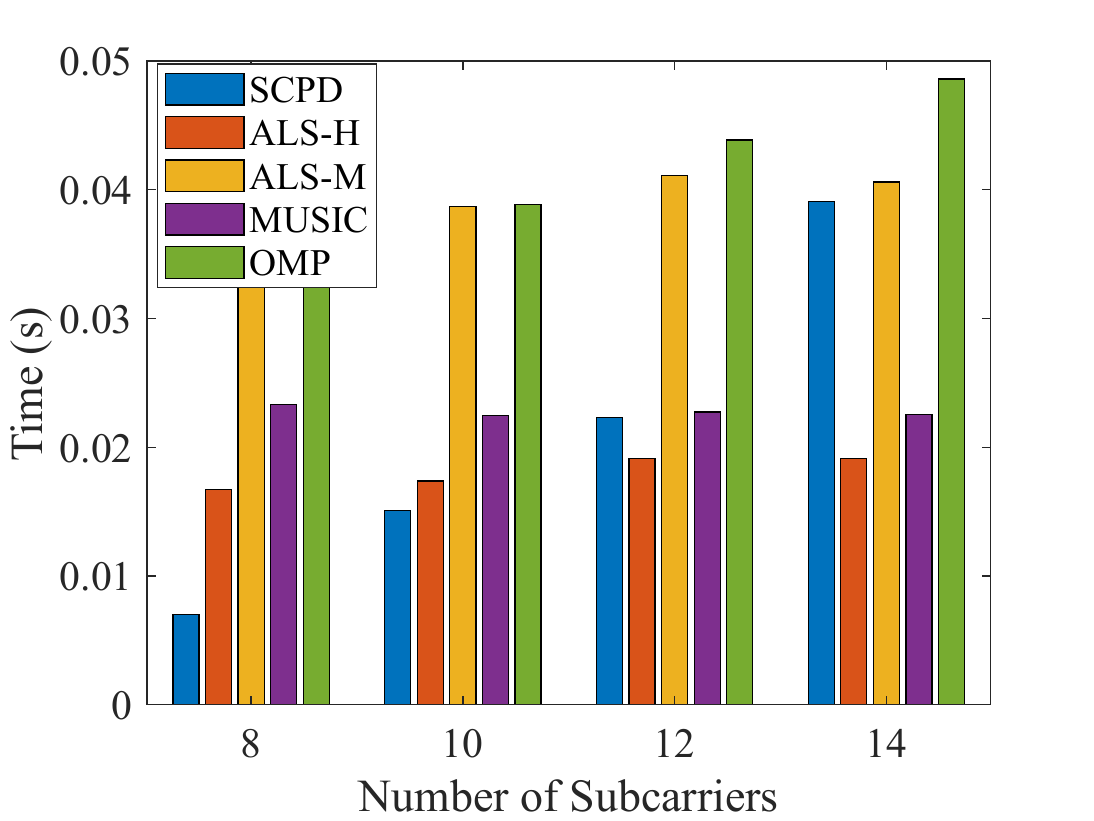}
    \caption{CPU time of different algorithms.}
    \label{fig:placeholder}
\end{figure}

\section{Conclusion}
In this paper, a tensor decomposition-based dynamic channel estimation algorithm for mmWave MA MIMO systems was proposed. Specifically, dynamic channel was  {modeled as a superposition of the channel contributions from the sparse propagation paths} based on the field response model and received signal was reformulated into a  {fourth-order} tensor to address the high-dimensional structure. Then,  {the ALS and SCPD algorithms} were utilized to solve tensor decomposition problem and uniqueness of the decomposition in our model was proved. Subsequently, CRB for each parameter was derived in detail and the parameters were extracted from  factor matrices and the channel matrix was rebuilt based on these parameters. Simulation results demonstrated that the proposed algorithms exhibited effectiveness when compared with the baseline algorithms, which verified the superiority of our proposed algorithm in accuracy. Finally, the complexity evaluation and the runtime result demonstrated that the proposed algorithm  {is less computationally demanding}. Consequently, the proposed tensor decomposition-based channel estimation algorithm offered significant practical advantages for implementation.   {The research presented in this paper can be extended to more realistic dynamic channel models that incorporate path birth and death, blockage, and time-varying angles and delays, which constitutes a promising direction for future research.}

\appendices
\section{Derivation of CRB}
\label{ccb}
{ 
Let $\widetilde{\mathbf b}_r$ denote the unit-gain delay factor, and
define the unit-gain vectorized contribution of the $r$th path as
\begin{equation}
\mathbf g_r
=
\mathbf d_r\otimes\mathbf c_r\otimes
\widetilde{\mathbf b}_r\otimes\mathbf a_r,
\qquad
\boldsymbol\mu(\boldsymbol\eta)
=
\sum_{r=1}^{R}\beta_r\mathbf g_r,
\label{eq:mean-vector-revised}
\end{equation}
where the Kronecker-product order follows the vectorization order of the
fourth-order received tensor. The real-valued parameter vector is defined as
\begin{equation}
\boldsymbol\eta
=
\left[
\boldsymbol\theta^T,\,
\boldsymbol\phi^T,\,
\boldsymbol\tau^T,\,
\boldsymbol\nu^T,\,
\operatorname{Re}\{\boldsymbol\beta\}^T,\,
\operatorname{Im}\{\boldsymbol\beta\}^T
\right]^T
\in\mathbb{R}^{6R\times 1}.
\label{eq:real-parameter-vector-revised}
\end{equation}

For the delay and Doppler factors, let
\begin{align}
[\widetilde{\mathbf b}_r]_k
&=
\exp\!\left\{-j2\pi\Delta_f p_k\tau_r\right\},\\
\frac{\partial[\widetilde{\mathbf b}_r]_k}{\partial\tau_r}
&=
-j2\pi\Delta_f p_k[\widetilde{\mathbf b}_r]_k,\\
[\mathbf c_r]_m
&=
\exp\!\left\{j2\pi\nu_r(m-1)N_sT_s\right\},\\
\frac{\partial[\mathbf c_r]_m}{\partial\nu_r}
&=
j2\pi(m-1)N_sT_s[\mathbf c_r]_m,
\label{eq:factor-derivatives-revised}
\end{align}
where $p_k$ is the index of the $k$th pilot subcarrier. Consequently,
\begin{align}
\frac{\partial\boldsymbol\mu}{\partial\theta_r}
&=
\beta_r\left(
\mathbf d_r\otimes\mathbf c_r\otimes
\widetilde{\mathbf b}_r\otimes
\frac{\partial\mathbf a_r}{\partial\theta_r}
\right),\\
\frac{\partial\boldsymbol\mu}{\partial\phi_r}
&=
\beta_r\left(
\frac{\partial\mathbf d_r}{\partial\phi_r}\otimes
\mathbf c_r\otimes\widetilde{\mathbf b}_r\otimes\mathbf a_r
\right),\\
\frac{\partial\boldsymbol\mu}{\partial\tau_r}
&=
\beta_r\left(
\mathbf d_r\otimes\mathbf c_r\otimes
\frac{\partial\widetilde{\mathbf b}_r}{\partial\tau_r}
\otimes\mathbf a_r
\right),\\
\frac{\partial\boldsymbol\mu}{\partial\nu_r}
&=
\beta_r\left(
\mathbf d_r\otimes
\frac{\partial\mathbf c_r}{\partial\nu_r}\otimes
\widetilde{\mathbf b}_r\otimes\mathbf a_r
\right).
\label{eq:physical-mean-derivative}
\end{align}
Most importantly, the derivatives with respect to the two real degrees of
freedom of the complex path gain are
\begin{equation}
\frac{\partial\boldsymbol\mu}
{\partial\operatorname{Re}\{\beta_r\}}
=
\mathbf g_r,
\qquad
\frac{\partial\boldsymbol\mu}
{\partial\operatorname{Im}\{\beta_r\}}
=
j\mathbf g_r.
\label{eq:complex-gain-derivatives-revised}
\end{equation}
Stacking the $\text{6}R$ derivative vectors in the same order as
$\boldsymbol\eta$ gives the Jacobian matrix $\mathbf G_\eta$. For
parameter-independent proper complex Gaussian noise with covariance
$\mathbf C_w$
\begin{equation}
\mathbf C_w
=
\sigma_w^2\mathbf I_{Q_{\rm MS}N_sKM}, \quad \sigma_w^2
\triangleq
\mathbb E\!\left\{
\left|[\mathcal N^{\mathbf Q}]_{q,s,k,m}\right|^2
\right\}.
\end{equation}
The real FIM and the corresponding CRB can be given as
\begin{equation}
\mathbf J(\boldsymbol\eta)
=
2\operatorname{Re}\!\left\{
\mathbf G_\eta^H\mathbf C_w^{-1}\mathbf G_\eta
\right\},
\qquad
\operatorname{CRB}(\boldsymbol\eta)
=
\mathbf J(\boldsymbol\eta)^{-1},
\label{eq:compact-real-fim-revised}
\end{equation}
provided that $\mathbf J(\boldsymbol\eta)$ is nonsingular. The complex
Gram matrix $\mathbf G_\eta^H\mathbf C_w^{-1}\mathbf G_\eta$ is Hermitian
positive semidefinite, whereas $\mathbf J(\boldsymbol\eta)$ is real,
symmetric, and positive semidefinite.
}
\balance
\bibliographystyle{ieeetr}
\bibliography{ref}

\end{document}